\documentclass[%
reprint,
superscriptaddress,
showpacs,
nofootinbib,
 amsmath,amssymb,
 aps,
pra
]{revtex4-1}

\usepackage{graphicx}
\usepackage{dcolumn}
\usepackage{bm}
\usepackage{hyperref}

\usepackage{braket}
\usepackage{amsthm}
\usepackage{epstopdf}
\usepackage{subfigure}
\usepackage{enumerate}

\newcommand{\aref}[1]{\hyperref[#1]{Appendix~\ref{#1}}}

\newtheorem{theorem}{Theorem}

\begin{document}

\preprint{APS/123-QED}

\title{Quantum centrality testing on directed graphs via PT-symmetric quantum walks}

\author{J. A. Izaac}
\email{josh.izaac@uwa.edu.au}
\author{J. B. Wang}
\email{jingbo.wang@uwa.edu.au}
\author{P. C. Abbott}
\affiliation{School of Physics, The University of Western Australia, Crawley WA 6009} \author{X. S. Ma}
\affiliation{School of Physics, Nanjing University, Nanjing, 210093, China}

\date{\today}

\begin{abstract}
Various quantum-walk based algorithms have been proposed to analyse and rank the centrality of graph vertices. However, issues arise when working with directed graphs --- the resulting non-Hermitian Hamiltonian leads to non-unitary dynamics, and the total probability of the quantum walker is no longer conserved. In this paper, we discuss a method for simulating directed graphs using PT-symmetric quantum walks, allowing probability conserving non-unitary evolution. This method is equivalent to mapping the directed graph to an undirected, yet weighted, complete graph over the same vertex set, and can be extended to cover interdependent networks of directed graphs. Previous work has shown centrality measures based on the CTQW provide an eigenvector-like quantum centrality; using the PT-symmetric framework, we extend these centrality algorithms to directed graphs with a significantly reduced Hilbert space compared to previous proposals. In certain cases, this centrality measure provides an advantage over classical algorithms used in network analysis, for example by breaking vertex rank degeneracy. Finally, we perform a statistical analysis over ensembles of random graphs, and show strong agreement with the classical PageRank measure on directed acyclic graphs.
\end{abstract}

\maketitle


\section{Introduction}
Quantum walks, the quantum analogue of the classical random walk, were first introduced by \citet{aharonov1993} in his seminal paper. Since then, they have become an important tool in the field of quantum computation and quantum information theory --- enabling the creation of quantum algorithms faster than their classical counterparts \cite{wang2013,berry2010,berry2011,childs2004,douglas2008,childs2003,rudinger2013,mahasinghe2015,paparo2013}, finding use in modelling complex quantum dynamical systems (such as photosynthesis and electron transport \cite{sarovar2010,engel2007,mohseni2008,rebentrost2009}), and providing a universal method of quantum computation \cite{childs2009,childs2013}. This is due, in part, to the markedly differing behaviour of the quantum walk; evolving as per the Schr\"odinger equation, the quantum walk is a time-reversible rather than a diffusive Markovian process. Coupled with inherent quantum effects such as superposition, interference and quantum correlations, quantum walks have therefore become an integral tool in linking network analysis and modelling with the problem solving potential of quantum computation.

One disadvantage, however, of the quantum walk, is the condition of unitarity due to the quantum nature of the walkers. As such, the standard quantum walk is unable to model or analyse directed network structures, without either a) resulting in non-unitary dynamics, or b) modifying the framework. This serves as a particular hindrance in extending established quantum algorithms (e.g. quantum search, centrality measures, graph isomorphism) and quantum dynamical models (such as transport of electrons or excitons) to systems with directions or biased potentials. Compare this to a classical random walk, where as long as the transition matrix remains stochastic, directed networks pose no problem. Consequently, various work-arounds have been proposed for dealing with this non-unitary behaviour, for example Szegedy quantum walks \cite{paparo2013,loke2017} and open-quantum walks \cite{sinayskiy2013,loke2017}. Unfortunately, there are numerous downsides to these approaches --- each requiring a significantly expanded Hilbert space, whilst not guaranteeing probability conservation of the system under study or muting the effect of quantum behaviour (due to environmental dephasing and loss in the open-quantum walk).

An alternative solution to this issue arrives in the form of PT-symmetry. First discovered by \citet{bender1998}, PT-symmetry started off as a simple curiosity --- the appearance of non-Hermitian Hamiltonians that exhibit a real eigenspectra, thus allowing non-unitary probability conservation via a redefinition of the Hilbert space inner product. This was attributed to parity-time symmetry of the non-Hermitian Hamiltonian, and over time was generalised to allow for non-Hermitian Hamiltonian symmetry under a combination of \textit{any} linear and non-linear operators \cite{mostafazadeh2002,mostafazadeh2002-a,bender1998,mannheim2013,li2014}. Simultaneously, PT-symmetric Hamiltonians were finding use in theoretical models of observed and predicted phenomena in condensed matter \cite{basu-mallick2000}, quantum field theory \cite{bender2004,bender2005}, and being observed and implemented in numerous optical experiments \cite{luo2013,peng2014,regensburger2012,rueter2010}; whilst older non-Hermitian studies in condensed matter \cite{dyson1956} and nuclear physics \cite{faisal1981} have been reformulated in the PT-symmetry framework \cite{jones-smith2013}. Recently, PT-symmetry has been used to model directed one-dimensional discrete-time quantum walks (DTQWs) \cite{mochizuki2016}, and considered in the case of continuous-time quantum walks (CTQWs) \cite{salimi2010}.

In this work, we present a rigorous framework for PT-symmetric CTQWs, extending the formalism of \citet{salimi2010} to ensure the initial quantum state vector is preserved. This is then broadened to the cases of multi-particle systems and interdependent networks. By expanding on the work of \cite{izaac2017}, in which the CTQW was shown to provide an eigenvector-like centrality measure for the quantum realm, we are then able to utilize the PT-symmetric CTQW to measure vertex centrality in several directed graphs examples. We show that our centrality scheme compares well to the classical PageRank algorithm \cite{brin1998} in the case of directed acyclic graphs (DAGs) and in some cases even breaks the vertex rank degeneracy characterized by the PageRank.

Unlike previous quantum centrality measures~\cite{paparo2013,sinayskiy2013,loke2017}, for a graph of $N$ vertices this scheme requires a Hilbert space of dimension $N$ (compare this to the Szegedy quantum walk based PageRank scheme, which requires a $N^2$ dimensional Hilbert space), and without the classical decoherence required for open quantum systems. Furthermore, we show that this formalism is equivalent to considering an undirected, yet weighted, complete graph with self-loops, providing a structural interpretation that may lead to simple experimental implementation.

This paper is structured as follows. In \autoref{sec:ctqw}, we detail the standard CTQW, and briefly discuss the issues with extending this to directed graphs. A brief introduction to the PT-symmetry and pseudo-Hermitian formalism is presented in \autoref{sec:pt}, tailored towards a quantum walker application. In \autoref{sec:phctqw}, we introduce our framework for PT-symmetric continuous-time quantum walks, and discuss how this can be interpreted as a mapping to a weighted and undirected CTQW on a complete graph; the framework is then extended to multi-particle systems and interdependent networks. A potential application, network centrality, is presented in \autoref{sec:centrality}, and a statistical analysis performed to verify its performance over an ensemble of randomly generated Erd\H{o}s-R\'enyi and scale-free graphs. Finally, our conclusions are provided in \autoref{sec:conc}.

\section{Continuous-time quantum walks}
\label{sec:ctqw}
Consider an arbitrary undirected graph $G(V,E)$, composed of vertices $j\in V$ and edges $(i,j)\in E$. The adjacency matrix of $G$ is then defined as
\begin{align}
	A_{ij} = \begin{cases}
    	1, & (i,j)\in E\\
        0, & (i,j)\notin E ~.
    \end{cases}
\end{align}
For a continuous-time quantum walk on graph $G$, the Hamiltonian is given by the discrete Laplacian (or graph Laplacian) of the adjacency matrix:
\begin{align}\label{eq:hamiltonian}
	H_{ij} = \left(\sum_k A_{ik}\right)\delta_{ij} -A_{ij}~.
\end{align}
To find the time-evolution of the walker, we solve the Schr\"odinger equation, assuming for simplicity atomic units ($m_e=e=\hbar=1$),
\begin{align}
i \frac{d}{dt} \ket{\psi(t)} = H\ket{\psi(t)},
\end{align}
which has the formal solution
\begin{align}
	\ket{\psi(t)} = e^{-iHt}\ket{\psi(0)}.
\end{align}
Note that the complex valued state vector $\ket{\psi(t)}=\sum_j a_j (t) \ket{j}$, where $a_j(t)=\braket{j|\psi(t)}\in\mathbb{C}$, represents the probability amplitude of the walker being found at node $j$ at time $t$, with $|a_j(t)|^2=|\braket{j|\psi(t)}|^2$ the resulting probability. Unlike discrete-time formulations of the quantum walk, in which probability amplitudes can only transition between adjacent (or `local') vertices at each time-step, the CTQW is a \textit{global} process. That is, it is possible for probability amplitude to transition to non-adjacent vertices in the graph at each infinitesimal time-step $\Delta t$.

As the adjacency matrix $A$ is real and symmetric, it is a Hermitian matrix, and it is easy to see from \autoref{eq:hamiltonian} that $H$ must also be Hermitian. It therefore follows that the time evolution operator, $U=e^{-iHt}$, is unitary ($UU^\dagger=I$), guaranteeing that the norm of $\ket{\psi(t)}$ is conserved under a continuous-time quantum walk, as required.

Let's now modify $G$ so that it is a \textit{directed} graph --- that is, the edgeset $(i,j)\in E$ is now described by an \textit{ordered} pair of vertices. As the adjacency matrix is no longer symmetric ($A_{ij}\neq A_{ji}~~ \forall i,j$), the Hamiltonian $H$ is no longer Hermitian, and thus we no longer have unitary time evolution ($UU^\dagger\neq I$). As a consequence, the norm squared of the quantum state is no longer conserved,
\begin{align}
\braket{\psi(t)|\psi(t)} = \braket{\psi(0)|U^\dagger U|\psi(0)} \neq \braket{\psi(0)|\psi(0)}
\end{align}
and may in fact grow or decay exponentially. Various modifications proposed for dealing with this non-unitary behaviour (for example, Szegedy quantum walks \cite{paparo2013,loke2017} and open-quantum walks \cite{sinayskiy2013,loke2017}) require a significantly expanded Hilbert space, resulting in considerable resource overhead in physical implementation.

\section{PT-symmetry}
\label{sec:pt}
Whilst non-Hermitian Hamiltonians with complex eigenvalues result in exponentially growing or decaying time-evolution, a wide variety of non-Hermitian Hamiltonians have been found to possess real eigenvalue spectra.
It was first noted by \citet{bender1998} that particular non-Hermitian Hamiltonians with real spectra exhibited PT-symmetry, that is,
\begin{align}
	[H,\mathcal{PT}]=0
\end{align}
where $\mathcal{P}: (\hat{x},\hat{p})\rightarrow(-\hat{x},-\hat{p})$ is the parity transformation operator, and $\mathcal{T}: (\hat{x},\hat{p})\rightarrow(\hat{x},-\hat{p})$ the  time reflection operator satisfying $\{\mathcal{T},i\}=0$ (anti-linearity) \cite{bender1998,fernandez1999}. On the basis of this observation, it was posited that invariance of a Hamiltonian under PT-transformations provides a more general condition for the reality of eigenspectra than simply Hermiticity. Immediately, research into PT-symmetric Hamiltonians found it was not so clear-cut; due to the anti-linearity of the $\mathcal{PT}$ operator, a $\mathcal{PT}$ invariant Hamiltonian may still undergo spontaneous symmetry breaking, leading to complex conjugate pairs of eigenvalues \cite{weigert2003,demorissonfaria2006}. Furthermore, although the existence of PT-symmetry is a sufficient condition for real spectra, it is not necessary. This same property can be found in Hamiltonians not exhibiting PT-symmetry --- thus failing to account for the existence of \textit{all} non-Hermitian Hamiltonians with real eigenspectra.

An alternative framework was put forward by \citet{mostafazadeh2002}. Denoted pseudo-Hermiticity, it was shown that for all diagonalizable non-Hermitian Hamiltonians exhibiting a real eigenspectra, there exists a positive semidefinite linear operator $V=\eta^\dagger\eta$ such that
\begin{align}\label{eq:pseudoHNew}
  H^{\dagger}=V H V^{-1}.
\end{align}
Additionally, it was proven that every PT-symmetric and diagonalizable Hamiltonian is pseudo-Hermitian \cite{mostafazadeh2002}. Coupled with the fact that $V=I$ corresponds to the case of Hermitian $H$, it was claimed that the pseudo-Hermiticity framework is the correct generalization of Hermiticity to non-Hermitian Hamiltonians.

Subsequent research has further explored the connections and similarities between PT-symmetry and pseudo-Hermiticity \cite{mostafazadeh2002-a}, with a flurry of papers released proclaiming the supremacy of one or the other, in the constant struggle to be seen as the more general of the two. Of particular note, it has been shown that even though the pseudo-Hermitian similarity transform is decidedly linear, pseudo-Hermiticity is a necessary and sufficient condition for $H$ to admit anti-linear symmetry \cite{mostafazadeh2002-b,solombrino2002,scolarici2003}; i.e. the condition of pseudo-Hermiticity is equivalent to the condition $[H,\Omega]=0$, where $\Omega$ is an anti-linear invertible or involutory operator. On this basis, one can conclude that any time-reversal invariant Hamiltonian belongs to the class of pseudo-Hermitian Hamiltonians, although the converse is not true, as $\Omega$ is not always guaranteed to be $\mathcal{T}$.

These two competing frameworks were finally reconciled by \citet{bender2010}, who introduced the concept of \textit{generalized} PT-symmetry; here, $\mathcal{P}$ represents \textit{any} linear operator (not just parity), and likewise $\mathcal{T}$ represents \textit{any} anti-linear operator; the chosen operators $\mathcal{P}$ and $\mathcal{T}$ need not commute. This generalized PT-symmetry condition is necessary and sufficient for reality of the characteristic equation,
\begin{align}
	|H-\lambda I|=0
\end{align}
which results in real eigenvalues if $\mathcal{PT}$ and $H$ are simultaneously diagonalizable, and complex conjugate pairs if not\footnote{This is an example of spontaneous PT-symmetry breaking. Even though a Hamiltonian may display $\mathcal{PT}$ invariance (i.e. $[H,\mathcal{PT}]=0$), the eigenstates of $H$, denoted $\ket{\phi_n}$, are not necessarily simultaneously eigenstates of $\mathcal{PT}$, due to the antilinearity of the $\mathcal{PT}$ operator. If this is the case, then the eigenspectrum is composed of complex conjugate pairs of eigenvalues, and $\mathcal{PT}\ket{\phi_n}$ provides the eigenstates of $\mathcal{PT}$.}. Thus, a Hamiltonian with a real eigenspectra necessarily displays (generalized) PT-symmetry, regardless of its diagonalizability --- providing the generalisation of Hermiticity so sought after in the original parity-time and pseudo-Hermiticity frameworks. From hereon in, use of the term `PT-symmetry' will refer to generalized PT-symmetry.

Under this new, more general, framework, pseudo-Hermiticity exists as a subset of PT-symmetry \cite{bender2010,mannheim2013}, and has been expanded to include cases where $H$ is non-diagonalizable \cite{scolarici2003} --- in such cases, it is no longer possible to satisfy the pseudo-Hermiticity similarity transform with a linear operator $V$ that is positive semidefinite --- leading to spontaneous symmetry breaking and complex conjugate pairs of eigenvalues. Moreover, \citet{bender2010} provides a criteria to determine whether a positive semi-definite $V$ exists for known PT-symmetric Hamiltonians\footnote{If $[\mathcal{C},\mathcal{PT}]=0~\forall~\mathcal{C}$ s.t. $\mathcal{C}^2=1$ and $[\mathcal{C},H]=0$, then there exists a positive semidefinite linear operator $V=\mathcal{CP}$ s.t. $VHV^{-1}=H^\dagger$}. In the following section, we will briefly outline the pseudo-Hermitian operator framework, and provide a method for determining $V=\eta^\dagger\eta$ in cases where there is no spontaneous PT-symmetry breaking.

Let $H$ be a non-Hermitian matrix. It is pseudo-Hermitian (and thus PT-symmetric), if it is related by a similarity transform to a Hermitian matrix $\tilde{H}$,
\begin{align}\label{eq:pseudoH}
  \tilde{H}=\eta H\eta^{-1},
\end{align}
where $\eta$ is frequently referred to in the literature as the \textit{pseudo-Hermitian operator} or \textit{metric}. Without loss of generality, we assume $\eta$ is an Hermitian operator ($\eta=\eta^\dagger$).  Due to the properties of a similarity transform, the eigenvalues of $H$ will be the same as $\tilde{H}$ and necessarily real. Taking the conjugate transpose of this result, we get
\begin{align}
  \tilde{H}^\dagger = \left(\eta^{-1}\right)^\dagger H^\dagger \eta^\dagger = \eta^{-1} H^\dagger \eta.
\end{align}
Since $\tilde{H}$ is Hermitian, $\eta^{-1} H^\dagger \eta=\eta H\eta^{-1}$, and thus a pseudo-Hermitian matrix must satisfy the following similarity transform with its conjugate transpose:
\begin{align}\label{eq:PHR}
	H^\dagger = \eta^2 H \eta^{-2}.
\end{align}
Rewriting this in the form $\eta^2 H = H^\dagger \eta^{2}$, note that the right-hand side is simply the Hermitian conjugate of the left-hand side. This suggests that the following redefinition of the inner product,
\begin{align}\label{eq:etaprod}
	\braket{\cdots|\cdots}_\eta := \braket{\cdots|\eta^2|\cdots},
\end{align}
should be sufficient to conserve the systems probability. Indeed, by using the Schr\"odinger equation, we see that this is in fact the case when working with pseudo-Hermitian operators:
\begin{align*}
	\frac{d}{dt}\braket{\psi(t)|\psi(t)}_\eta &= \braket{\frac{d}{dt}\psi(t)|\eta^2|\psi(t)} + \braket{\psi(t)|\eta^2|\frac{d}{dt}\psi(t)}\notag\\
    & =  \braket{\psi(t)|iH^\dagger\eta^2|\psi(t)} - \braket{\psi(t)|\eta^2 iH|\psi(t)}\notag\\
    & =  i\braket{\psi(t)|\left(H^\dagger\eta^2-\eta^2 H\right)|\psi(t)}\notag\\
    & = 0
\end{align*}

As Hermitian matrices are always diagonalisable by their unitary eigenbasis ($\tilde{H}=P\Lambda P^\dagger$ where $P^{-1}=P^\dagger$), it follows from the similarity relation \autoref{eq:pseudoH} that pseudo-Hermitian matrices must also be diagonalisable\footnote{In fact, solving this equation allows you to find an expression for $\eta$ in terms of the left eigenvectors of $H$, $\bra{\phi_j}$, and the eigenvectors of $\tilde{H}$, $\ket{\tilde{\psi}_j}$: $\eta_{ij}=\sum_k \ket{\tilde{\psi}_k}_i \bra{\phi_j}_k$}:
\begin{align*}
	H = \eta^{-1}\tilde{H}\eta = \left(\eta^{-1}P\right)\Lambda\left(P^\dagger\eta\right) = \left(P^\dagger \eta\right)^{-1}\Lambda\left(P^\dagger\eta\right).
\end{align*}
Diagonalisable matrices must admit a biorthonormal eigenbasis \cite{anthony2012},
\begin{align}
	&H\ket{\psi_j}=\lambda_j\ket{\psi_j},\\
	&H^\dagger\ket{\phi_j}=\lambda_j\ket{\phi_j}, ~~ j=1,2,\dots,n
\end{align}
where $\braket{\phi_i|\psi_j}=\delta_{ij}$ and $\lambda_j\in\mathbb{R}$ due to pseudo-Hermiticity. The completeness relation is given by
\begin{align}
	I = \sum_j \ket{\psi_j}\bra{\phi_j}.
\end{align}

By applying the pseudo-Hermiticity relation (\autoref{eq:PHR}) to the biorthonormal eigenvector equations, we can deduce a method of constructing $\eta$. For instance,
\begin{align}
	H^\dagger \ket{\phi_j} = \eta^2H\eta^{-2}\ket{\phi_j} = \lambda_j\ket{\phi_j}.
\end{align}
Pre-multiplying both sides by $\eta^{-2}$,
\begin{align}
	H\left(\eta^{-2}\ket{\phi_j}\right) = \lambda_j\left(\eta^{-2}\ket{\phi_j}\right),
\end{align}
it can be seen that
\begin{align}
	\ket{\psi_j} = \eta^{-2}\ket{\phi_j}~~ \Leftrightarrow ~~ \ket{\phi_j} = \eta^2 \ket{\psi_j}.
\end{align}
Hence, $\eta$ acts to transform between the pseudo-Hermitian and Hermitian basis, and $\eta^2$ acts to transform between the pseudo-Hermitian biorthonormal eigenbasis. It therefore follows that the biorthonormal eigenbasis is the basis for the inner product space defined by \autoref{eq:etaprod}. Combining this result with the biorthonormal completeness relation, we arrive at a method of constructing the pseudo-Hermitian operator $\eta$:
\begin{align}\label{eq:defineeta}
	\eta = \sqrt{\eta^2 \sum_j \ket{\psi_j}\bra{\phi_j}} = \sqrt{\sum_j \ket{\phi_j}\bra{\phi_j}}
\end{align}
and similarly
\begin{align}\label{eq:defineetainv}
	\eta^{-1} = \sqrt{ \sum_j \ket{\psi_j}\bra{\phi_j}\eta^{-2}} = \sqrt{\sum_j \ket{\psi_j}\bra{\psi_j}}.
\end{align}
Since we have defined $V=\eta^2$, it follows that
\begin{align}
	V=\sum_j \ket{\phi_j}\bra{\phi_j}~~\Leftrightarrow~~~V^{-1}=\sum_j \ket{\psi_j}\bra{\psi_j}.
\end{align}
Further, as $V$ is positive semidefinite, we can be assured that the square root function applied to the operators in \autoref{eq:defineeta} and \autoref{eq:defineetainv} is well defined, albeit admitting multiple solutions. For consistency, we will choose $\eta$ to be the \textit{principle square root} --- the unique positive semidefinite square root of $V$.

So, to briefly summarise, $H$ is necessarily pseudo-Hermitian (and consequently PT-symmetric) if it satisfies \textbf{any one} of the following equivalent conditions:
\begin{enumerate}
	\item \textbf{$H$ is similar to a Hermitian matrix}. There exists a Hermitian operator $\eta$ and a Hermitian matrix $\tilde{H}$ such that $\tilde{H}=\eta H\eta^{-1}$.
	\item \textbf{$H$ is similar to its own Hermitian conjugate}. There exists a positive Hermitian operator $\eta$ such that $H^\dagger=\eta^2 H\eta^{-2}$.
	\item \textbf{$H$ has real eigenvalues and is diagonalizable} Note that a matrix is diagonalisable if and only if it has $n$ linearly-independent eigenvectors \cite{anthony2012}.
\end{enumerate}

\section{Pseudo-Hermitian continuous-time quantum walks}
\label{sec:phctqw}
First introduced by \citet{salimi2010}, pseudo-Hermitian continuous-time quantum walks take advantage of the pseudo-Hermitian structure of various graphs in order to implement directed quantum walks. In their study, the transition probability of the pseudo-Hermitian CTQW at vertex $j$ at time $t$ is defined to be
\begin{align}
	\mathbb{P}_j(t) = |\braket{j|\eta e^{-iHt}|\psi(0)}|^2.
\end{align}
Note that this does \textit{not} preserve the initial state;
\begin{align}
	\mathbb{P}_j(0) = |\braket{j|\eta|\psi(0)}|^2 \neq |\braket{j|\psi(0)}|^2.
\end{align}
Thus, if $\ket{\psi(0)}$ is chosen to be an equal superposition over all vertices, this will not be reflected in the quantum walk at time $t=0$, making this definition unsuitable for algorithms such as centrality testing and graph isomorphism. Thus, rather than utilise their implementation, we present an alternative formulation.

As our aim is to experimentally produce a pseudo-Hermitian CTQW for network analysis, rather than redefine the Hilbert space inner-product as per \autoref{eq:etaprod}, the inner product will not be modified. Instead, we have three options available:

\begin{enumerate}[A.]
	\item \textbf{No modification}: implement the time-evolution operator using the normal framework for the CTQW, $U=e^{-iHt}$, with no redefinition of the inner-product. This results in a non-unitary and non-probability conserving time-evolution, but the pseudo-Hermiticity of the system ensures the norm squared of the quantum walk wavefunction will just oscillate with no exponential growth and decay. 
    \item \textbf{Modify the time-evolution operator}: the non-unitary time evolution operator from (A) is instead modified as follows,
\begin{align}
	\tilde{U}(t) = \eta U(t)\eta^{-1} = \eta e^{-iHt} \eta^{-1}
\end{align}
where $\tilde{U}(t)$ is unitary due to the pseudo-Hermitian similarity transform of the matrix exponential. This reflects the underlying directional structure of the graph, whilst allowing for probability conservation. Note that whilst the product is unitary, $\eta$ and $U(t)$ are non-unitary matrices.
    \item \textbf{Modify the Hamiltonian}: in this approach, the pseudo-Hermitian Hamiltonian is modified via similarity transform to make it Hermitian,
  \begin{align}
      \tilde{U}(t) = e^{-i\tilde{H}t} = e^{-i\eta H\eta^{-1}t}
  \end{align}
  This preserves the directional structure of the graph, whilst allowing us to use the standard CTQW framework. Furthermore, it has the potential to be implemented experimentally via quantum simulation.
\end{enumerate}
Note that options (B) and (C) are equivalent  --- it is only their resulting experimental implementations which would differ. From hereon, the modified CTQW walk outlined in (B) and (C) will be referred to as the $\eta$-CTQW, to distinguish it from the standard non-unitary CTQW in (A).

\subsection{3-vertex directed graph}
Consider the 3-vertex graph in \autoref{fig:3vertex}. Its Hamiltonian is found by calculating the graph Laplacian,
\begin{align}
	H = \left[\begin{matrix}
    		1 & -1 & -1 \\
            -1 & 1 & -1 \\
            0 & 0 & 2
    	\end{matrix}\right]
\end{align}
Whilst not Hermitian ($H^\dagger\neq H$), the eigenvalues ($\lambda=0,2,2$) are all real and the eigenvectors are linearly independent; thus $H$ is PT-symmetric and pseudo-Hermitian.

\begin{figure}[htp]
	\centering
    \includegraphics[scale=0.7]{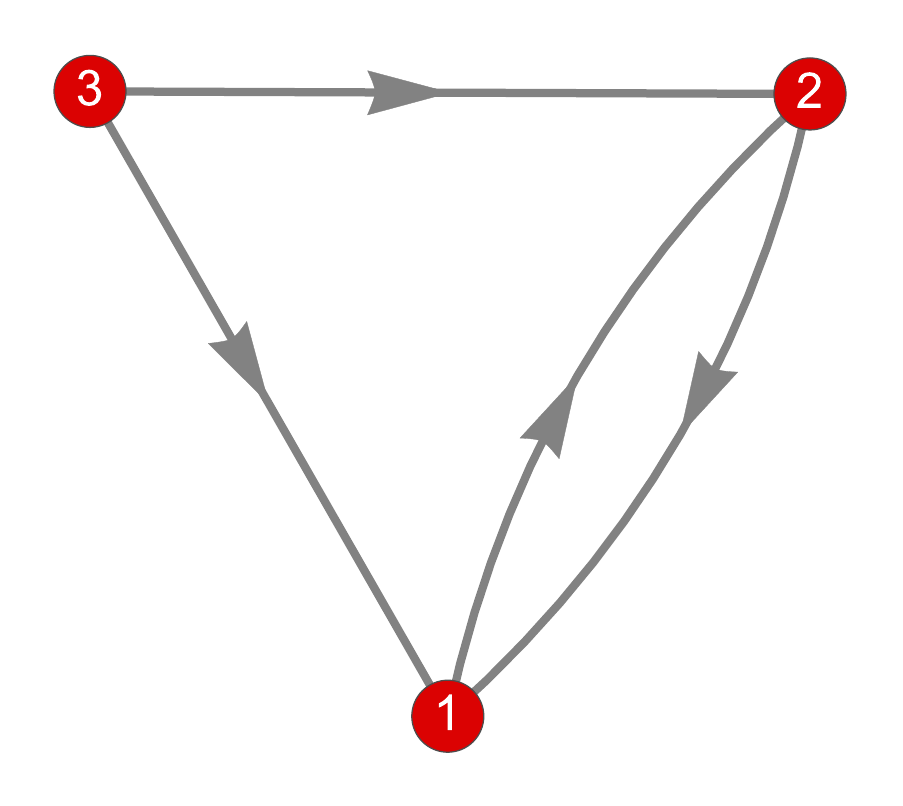}
    \caption{3-vertex directed graph}
    \label{fig:3vertex}
\end{figure}

\subsubsection{Standard CTQW (non-probability conserving)}
Solving the Schr\"odinger equation for the standard CTQW (\autoref{fig:3vertexPlot1}), we find that although there is no exponential growth or decay of the squared norm of the wavefunction (due to the pseudo-Hermiticity), nevertheless, the squared norm oscillates and is not conserved.
\begin{figure}[htp!]
	\centering
    \includegraphics[scale=0.75]{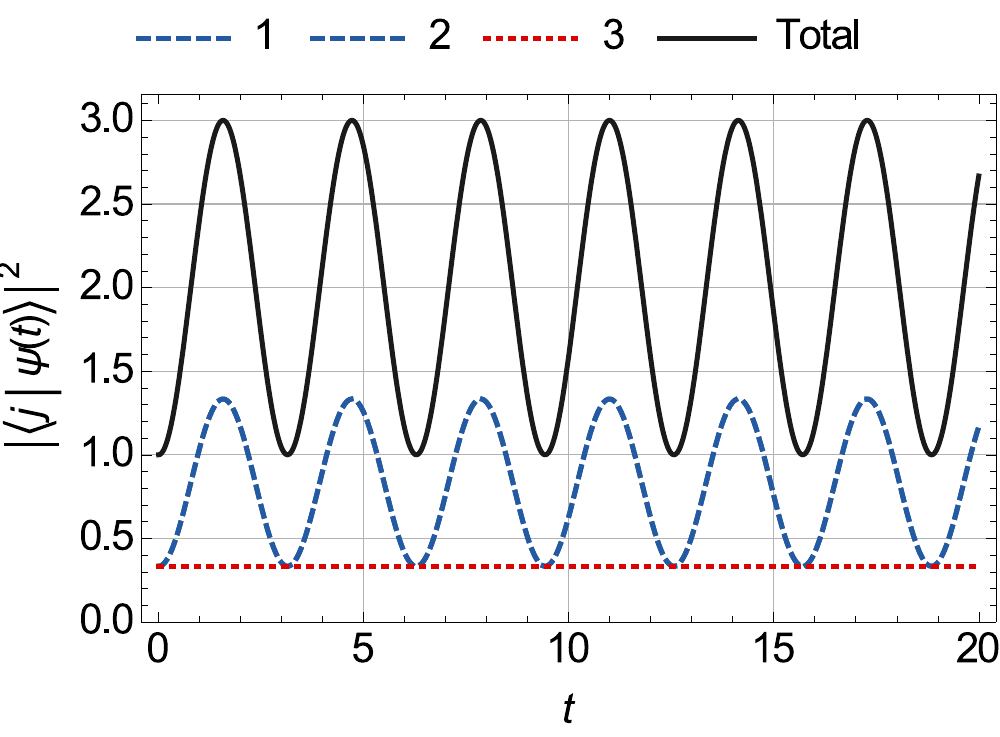}
    \caption{Vertex squared norm vs. time for the non-unitary CTQW on the directed graph \autoref{fig:3vertex}, with the walker initially in an equal superposition of vertex states $\ket{\psi(0)}=\frac{1}{\sqrt{3}}\sum_j \ket{j}$. Vertex 1 and 2 have equal norm squared (blue, dashed), with a higher time-average than vertex 3 (red, dotted). The squared norm of the total system is given by the black, solid line.}
    \label{fig:3vertexPlot1}
\end{figure}
\begin{figure}[htp!]
	\centering
    \includegraphics[scale=0.75]{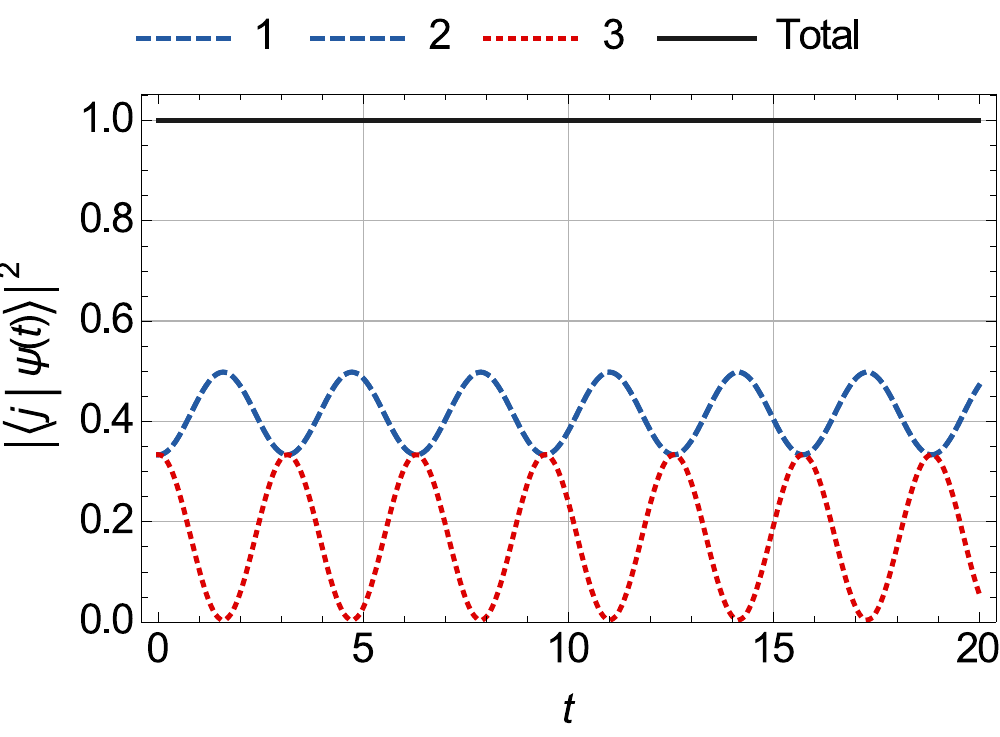}
    \caption{Vertex probability vs. time for the pseudo-Hermitian $\eta$-CTQW on the directed graph \autoref{fig:3vertex}, with the walker initially in an equal superposition of vertex states $\ket{\psi(0)}=\frac{1}{\sqrt{3}}\sum_j \ket{j}$. Vertex 1 and 2 have equal probability (blue, dashed), with a higher time-average than vertex 3 (red, dotted). The total probability of the system is given by the black, solid line.}
    \label{fig:3vertexPlot1eta}
\end{figure}

In order to experimentally implement this CTQW, we wish to decompose the time evolution operator such that the time dependence is restricted to a single diagonal unitary matrix. Diagonalising $H$ yields the following decomposition:
\begin{align}
	U(t) = \frac{1}{2}\left[\begin{matrix}
    	1 & -1 & -1\\
        1 & 0 & 1\\
        0 & 1 & 0
    \end{matrix}\right]\left[\begin{matrix}
    	1 & 0 & 0\\
        0 & e^{-2it} & 0\\
        0 & 0 & e^{-2it}
    \end{matrix}\right]\left[\begin{matrix}
    	1 & 1 & 1\\
        0 & 0 & 2\\
        -1 & 1 & -1
    \end{matrix}\right]
\end{align}
Note that whilst the diagonal matrix is unitary, the two outer matrices are non-unitary, resulting in a non-unitary time evolution operator $U(t)$.

\subsubsection{\texorpdfstring{$\eta$}{eta}-CTQW (probability conserving)}

By calculating the biorthonormal eigenbasis, we can derive the pseudo-Hermitian operator $\eta$ of the graph,
\begin{align}
	\eta= \frac{1}{6}\left[\begin{matrix}
    		3+2\sqrt{2} & -3+2\sqrt{2} & \sqrt{2}\\
    		-3+2\sqrt{2} & 3+2\sqrt{2} & \sqrt{2}\\
    		\sqrt{2} & \sqrt{2} & 5\sqrt{2}
    	\end{matrix}\right]
\end{align}
Thus, the modified time-evolution operator is given by $\tilde{U}(t) = \eta e^{-iHt} \eta^{-1}$ --- see 
\autoref{fig:3vertexPlot1eta} for how this affects the dynamics of the quantum walk.

In order to find a useful decomposition, we need to be able to diagonalise $\eta H \eta^{-1}$; i.e.
\begin{align}
	D = S^{\dagger} \eta H \eta^{-1} S
\end{align}
Once $S$ is calculated, by similarity it can also be used to diagonalise $\tilde{U}(t)$. This results in the following decomposition:

\begin{align}
	\tilde{U}(t) = S\left[\begin{matrix}
    	1 & 0 & 0\\
        0 & e^{-2it} & 0\\
        0 & 0 & e^{-2it}
    \end{matrix}\right]S^\dagger
\end{align}
where 
\begin{align}
	S=\frac{1}{3}\left[\begin{matrix}
         2 & -\frac{3}{\sqrt{2}} & -\frac{1}{\sqrt{2}} \\
         2 & \frac{3}{\sqrt{2}} & -\frac{1}{\sqrt{2}} \\
         1 & 0 & 2 \sqrt{2}
    \end{matrix}\right]
\end{align}

Note that, unlike the non-probability conserving case, here all three of the decomposed matrices are unitary.

\subsection{Alternative interpretation}
To get an understanding of the transformation from the non-probability conserving CTQW to the pseudo-Hermitian $\eta$-CTQW, let's have a look at the pseudo-Hermitian Hamiltonian for the 3-vertex graph discussed above,
\begin{align}\label{eq:3vertexEta}
	\tilde{H} = \eta H\eta^{-1} = \frac{1}{9}\left[
          \begin{matrix}
				10 & - 8 & - 4 \\
                - 8 & 10 & - 4 \\
                - 4 & - 4 & 16 
          \end{matrix}
        \right]
\end{align}
Describing a discrete graph structure, the $\eta$-Hamiltonian $\tilde{H}$ is symmetric, yet contains fractional quantities --- we may interpret $\tilde{H}$ as the Laplacian of a weighted, undirected complete graph with self-loops. (Note that this is just one valid interpretation, and is not unique). In doing so, it is convenient to use a more rigorous definition of the Laplacian more suited to weighted graphs.

The \textit{oriented incidence matrix} $M$ of an undirected graph $G(V,E)$, with vertex set $V=\{v_1,v_2,\dots,v_N\}$ and edge set $E=\{e_1,e_2,\dots,e_m\}$, is a $n\times m$ matrix associated with a particular orientation of the edges of $G$. That is, each edge $e_j$ is given a random direction. The oriented incidence matrix is therefore defined as follows:
\begin{align}
	M_{ij} = \begin{cases}
    	2, & \text{if edge $e_j$ is a self-loop incident on node $v_i$}\\
    	1, & \text{if edge $e_j$ is incident away from node $v_i$}\\
    	-1, & \text{if edge $e_j$ is incident towards node $v_i$}\\
    	0, & \text{if edge $e_j$ is not incident on node $v_i$}
    \end{cases}
\end{align}
Note that, unlike the standard incidence matrix, columns of the oriented incidence matrix not associated with self-loops must sum to zero.

Once we have calculated the oriented incidence matrix of a weighted graph, we can compute the $N\times N$ \textit{weighted Laplacian} using the following relationship:
\begin{align}
	L = MWM^T
\end{align}
where $W$ is a diagonal matrix containing the weights $\{w_{ij}\}$ associated with the edges $\{(v_i,v_j)\}$. For a complete graph over $N$ vertices with self-loops and arbitrary edge weighting, it turns out we can compute the weighted Laplacian directly,
\begin{align}
	L_{ij} = \begin{cases}
    	\sum_{k=1}^N w_{ik}+ 3w_{ii}, & i=j\\
    	-w_{ij}, & i\neq j
    \end{cases}
\end{align}
where $w_{ij}=w_{ji}$ and the size of the set $\{w_{ij}\}$ is $|E|=\frac{1}{2}N(N+1)$. Thus, if we have an $N\times N$ pseudo-Hamiltonian $\tilde{H}$, then solving $\tilde{H}=L_{ij}$ provides the edge weighting 
\begin{align}\label{eq:edgeweights}
	w_{ij} = \begin{cases}
    	\frac{1}{4}\sum_{k=1}^N \tilde{H}_{ik}, & i=j\\
        -\tilde{H}_{ij}, &i\neq j
    \end{cases}
\end{align}
That is, we can interpret the $\eta$-CTQW of a \textit{directed} $N$-vertex graph in terms of a standard CTQW on a \textit{undirected} complete graph with self-loops, and edge weights given by $w_{ij}$ above. The undirected edge weightings allows us to approximate the directed dynamics we require.

For example, let's return to the 3-vertex graph examined previously. Using the pseudo-Hamiltonian of the graph (\autoref{eq:3vertexEta}), we can find the edge weights $w_{ij}$ via \autoref{eq:edgeweights}. See \autoref{fig:ewgraph} for the result.
\begin{figure}[htp!]
	\centering
    \includegraphics[scale=0.75]{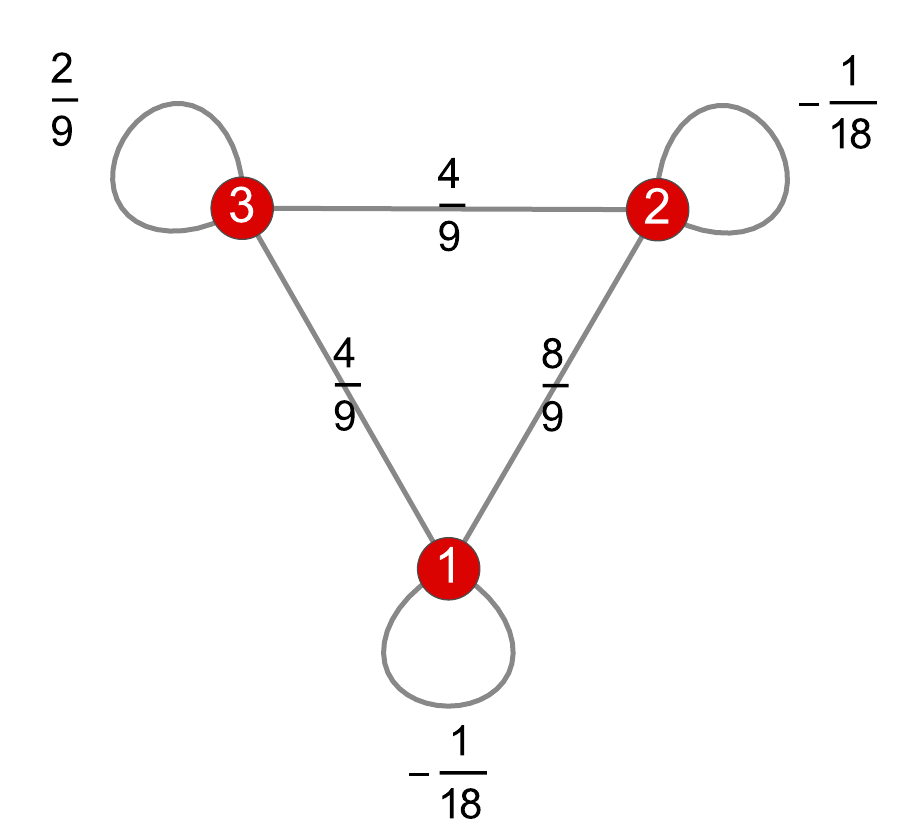}
    \caption{Performing an $\eta$-CTQW on the 3-vertex directed graph in \autoref{fig:3vertex} is equivalent to performing a standard CTQW on the undirected edge-weighted graph shown above.}
    \label{fig:ewgraph}
\end{figure}

\subsection{Multi-particle \texorpdfstring{$\eta$}{eta}-CTQW}
If we wish to extend the standard CTQW to simulate $P$ distinguishable particles on graph $G$, the Hamiltonian of the system is expanded to act on a $N^P$ Hilbert space, as follows:
\begin{align}
	H^{(P)} = H_1 \oplus H_2 \oplus \cdots \oplus H_P + \Gamma_{int}
\end{align}
where $H_j$ is the free-particle Hamiltonian of the $j$th particle on graph $G$, $\Gamma_{int}$ represents a potential interaction between the particles, and $\oplus$ is the tensor or Kronecker sum defined by
\begin{align}
	A_{n\times n} \oplus B_{m\times m} = 	A_{n\times n}\otimes I_{m\times m} + I_{n\times n}\otimes B_{m\times m}
\end{align}
Note that the free particle Hamiltonians are identical ($H_1=H_2=\cdots =H_P$), allowing this to be rewritten as 
\begin{align}
	H^{(P)} = H^{\oplus p}+\Gamma_{int}
\end{align}

In the case of no-interaction ($\Gamma_{int}=0$), then it is trivial to show that if $H$ is pseudo-Hermitian, then so is $H^{(P)}$, with pseudo-Hermitian operator $\eta^{(P)}$ given by
\begin{align}
	\eta^{(P)} = \underbrace{\eta\otimes\eta\otimes\dots\otimes\eta}_{P}=\eta^{\otimes P},
\end{align}
and $\eta$ the pseudo-Hermitian operator of the single particle Hamiltonian $H$.

When $\Gamma_{int}\neq 0$, then additional care must be taken --- unlike Hermitian matrices, the pseudo-Hermitian matrices are \textit{not} closed under addition, so $H^{(P)}$ is no longer guaranteed to be pseudo-Hermitian (this can easily be shown by counter-example).

\subsection{Interdependent networks}
Aside from the Cartesian product of graphs discussed above, another method of combining graph structures are interdependent networks. In the real-world, very few networks operate independently, instead interacting and depending on a myriad of other networks~\cite{gao2012, rosato2008} --- examples include modelling cascading failures between power grids, communication networks, and physiological/biochemical systems.
As such, being able to extend the pseudo-Hermitian CTQW to model interdependent networks greatly expands the scope of the framework. In this section, we will consider an interconnected network of two pseudo-Hermitian graphs and determine the properties that must be satisfied for the resulting network to have guaranteed pseudo-Hermiticity.

Let $A_1$ and $A_2$ refer to the adjacency matrices of two graphs, $G_1$ and $G_2$, respectively. The resulting interdependent network Hamiltonian is constructed as follows:
\begin{align}
	A = \left[\begin{matrix}
		A_1&B_0\\
		B_0^T&A_2
	\end{matrix}\right]
\end{align}
where $B_0$ is the adjacency matrix representing the edges connecting the two sets of vertices $V(A_1)$ and $V(A_2)$. The Hamiltonian of this graph, defined by \autoref{eq:hamiltonian}, can then be written as
\begin{align}
	H = \left[\begin{matrix}
		\mathcal{H}_1&-B_0\\
		-B_0^T&\mathcal{H}_2
	\end{matrix}\right]
\end{align}
where
\begin{align}
	(\mathcal{H}_u)_{ij}=(H_u)_{ij}+\sum_k (B_0)_{ik}\delta_{ij},~~~u\in\{1,2\}
\end{align}
and $H_1$ and $H_2$ are the Hamiltonians of $A_1$ and $A_2$ respectively.

\begin{theorem}
	The Hamiltonian of an interdependent network of two graphs $A_1$ and $A_2$, with pseudo-Hermitian Hamiltonians $H_1$ and $H_2$ respectively, will itself exhibit pseudo-Hermiticity if the inter-network connections $B_0$ is pseudo-Hermitian and degree-regular, and the commutation relations $H_1 B_0 = B_0 H_2$ and $H_2 B_0^T = B_0^T H_1$ are satisfied.
\end{theorem}

\begin{proof}
$B_0$ is degree regular with degree $c$ --- thus
\begin{align}
	\sum_{k}(B_0)_{ik}\delta_{ij}=c\delta_{ij} ~~\Rightarrow~~ \mathcal{H}_u = H_u + cI
\end{align}
If we decompose the Hamilton into the sum
\begin{align}
	H = \left[\begin{matrix}
			\mathcal{H}_1&0\\
			0&\mathcal{H}_2
		\end{matrix}\right] + \left[\begin{matrix}
				0&-B_0\\
				-B_0^T&0
			\end{matrix}\right] = A+B
\end{align}
we can now prove that the interdependent network Hamiltonian is always pseudo-Hermitian if the two components $A$ and $B$ commute and are themselves pseudo-Hermitian.

Since $H_1$ and $H_2$ are pseudo-Hermitian, they are therefore diagonalisable by matrices $Q_u$:
	\begin{align}
		\Lambda_u = Q_u^{-1} H_u Q_u
	\end{align}
where $\Lambda_u$ are real diagonal matrices of eigenvalues. It follows that
\begin{align}
	Q_u^{-1}\mathcal{H}_u Q_u = Q_u^{-1}(H_u+cI) Q_u = \Lambda_u +cI
\end{align}
That is, $\mathcal{H}_1$ and $\mathcal{H}_2$ are simultaneously diagonalisable and exhibit real eigenspectra --- satisfying the criteria for pseudo-Hermiticity. From here, it is trivial to show that $A$ is diagonalised as follows:
\begin{align}
	\Lambda_A = \left[\begin{matrix}
		\Lambda_1+cI& 0\\
		0&\Lambda_2+cI
	\end{matrix}\right] = Q_A^{-1} A Q_A
\end{align}
where
\begin{align}
	Q_A = \frac{1}{\sqrt{2}}\left[\begin{matrix}
		Q_1& 0\\
		0&Q_2
	\end{matrix}\right]
\end{align}
and $\Lambda_A$ is real.

As $B_0$ is also pseudo-Hermitian, therefore it is also diagonalisable with real eigenvalues; $\Lambda_{B_0}=Q_B^{-1} B_0 Q_B$. Similarly to above, we can use this result to diagonalise the matrix $B$:
\begin{align}
	\Lambda_B = \left[\begin{matrix}
		\Lambda_{B_0}& 0\\
		0&-\Lambda_{B_0}
	\end{matrix}\right] = Q_B^T B Q_B
\end{align}
where
\begin{align}
	Q_B = \frac{1}{\sqrt{2}}\left[\begin{matrix}
		Q_{B_0}& Q_{B_0}\\
		Q_{B_0}&-Q_{B_0}
	\end{matrix}\right]
\end{align}
and $\Lambda_B$ is real.

It follows from the above analysis that if $H_1$ and $H_2$ are pseudo-Hermitian, and the interconnections adjacency matrix $B_0$ is pseudo-Hermitian and degree regular, that the matrices $A$ and $B$ are also pseudo-Hermitian.

To ensure that the sum $A+B$ remains pseudo-Hermitian, we can make use of the well-known property that commutating diagonalisable matrices are simultaneously diagonalisable \cite{nehaniv2014}. This requires the additional constraint, $[A,B]=0$, resulting in the following two conditions that must be satisfied:
\begin{subequations}
\begin{align}
	H_1 B_0 &= B_0 H_2 \\
	H_2 B_0^T &= B_0^T H_1
\end{align}
\end{subequations}
If these are both satisfied, then $A$ and $B$ must be simultaneously diagonalised by a matrix $S$:
\begin{align}
	H = A+B = S(P_A\Lambda_A P_A^T + P_B\Lambda_B P_B^T) S^{-1}
\end{align}
where $P_A$ and $P_B$ are permutation matrices. From here, we can see that the eigenvalues of $H$ are contained in the set of elements $\{\lambda_1\pm\lambda_{B_0}+1\}\cup\{\lambda_2\pm\lambda_{B_0}+1\}$ and thus are necessarily real. Therefore, the interdependent network of pseudo-Hermitian graphs --- which commute with their interconnections --- also exhibits pseudo-Hermiticity.
\end{proof}

\begin{figure}[h!]
	\centering
	\includegraphics[scale=0.8]{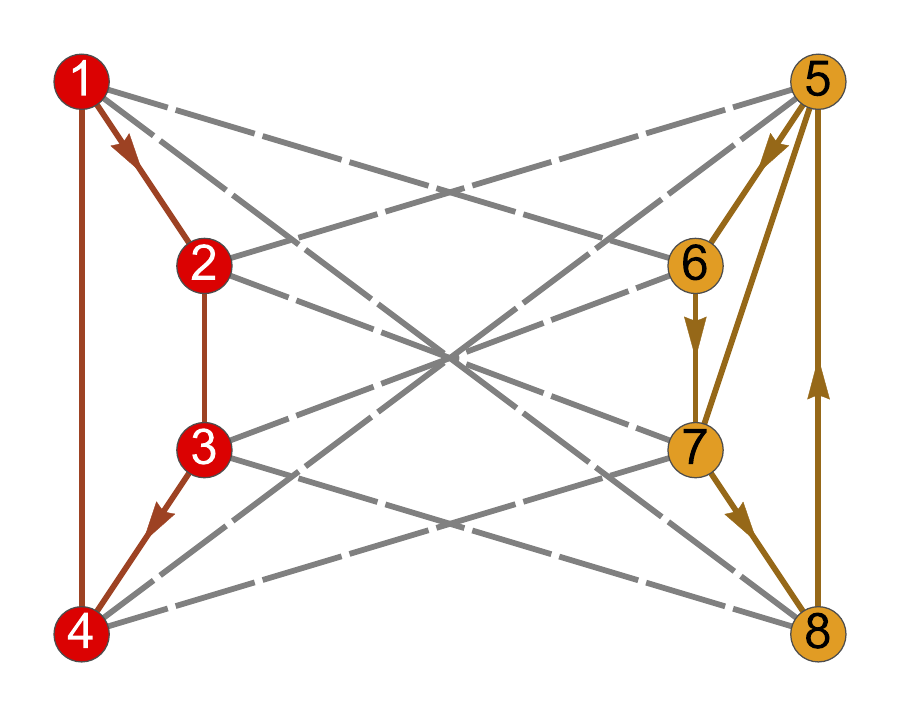}
	\caption{A pseudo-Hermitian interconnected network composed of two 4-vertex pseudo-Hermitian graphs (red and yellow respectively) connected via circulant interconnections (gray, dashed).}
	\label{fig:intnet}
\end{figure}

An example of a pseudo-Hermitian interconnected network in which the two graph Hamiltonians commute with the interconnections is given in \autoref{fig:intnet}. Note that the above theorem is not a necessary condition for interdependent network pseudo-Hermiticity --- examples can be constructed where $[A,B]\neq0$, yet $A+B$ remains diagonalisable and pseudo-Hermitian. Nevertheless, this results provides a useful method of constructing interdependent networks with guaranteed pseudo-Hermiticity.

\section{Centrality testing}
\label{sec:centrality}

\subsection{Centrality introduction}
In the study of network structure, centrality measures are an integral tool, allowing the determination and ranking of graph nodes deemed more important to the structure \cite{brin1998,knoke1983,brass1984,ibarra1993,ibarra1993-a,sade1989,amitai2004}. At its core, a graph centrality measure satisfies the following properties:
\begin{itemize}
\item $C:G(V,E)\rightarrow \mathbb{R}^{|V|}$ is a function or algorithm that accepts a graph as input, and returns a real valued, strictly positive vector over the set of vertices $V$. 
\item The highest values obtained correspond to vertices deemed more `important' or `central' to the graph structure.
\end{itemize}
However, with this general definition comes several caveats. Firstly, note that no meaning has been attributed to vertices with low centrality values --- this is deliberate, as noise grows successively larger for vertices beyond the topmost ranked vertices \cite{lawyer2015}. As such, centrality measures convey very little information regarding a majority of vertices; they are solely for determining the most central nodes (ranking \textit{all} the vertices is more the domain of \textit{influence measures} \cite{liu2013}). Secondly, what constitutes `importance' is subjective, and depends on the application or model to be analysed, and how information `flows' throughout the network \cite{borgatti2005}.

To account for the numerous ways we may quantify importance, various methods exist for measuring how central or important particular vertices are in a network structure; degree centrality, closeness centrality, eigenvector centrality, PageRank centrality, to name a few. It is pertinent to note that many of these centrality measures can be reformulated as classical random walks. For instance, the degree centrality is based on walks of length one, and is useful in models requiring direct and immediate influence between adjacent nodes. Moreover, it is trivial to show that the degree centrality is proportional to the limiting distribution of a random walk \cite{borgatti2006}. 

Eigenvector centrality, on the other hand, is a useful measure when considering long-term `indirect' influence between vertices; if a vertex with low degree is adjacent to a vertex with a high number of connections, the first vertex will likewise have a high eigenvector centrality measure \cite{borgatti2005}. Based on the spectral properties of a graph, the eigenvector centrality is defined by $C_j^{(ev)}=\mathbf{v}_j$, where $\mathbf{v}$ is the eigenvector of the adjacency matrix with maximum eigenvalue \cite{pastor-satorras2016}. Also referred to as the \textit{principle eigenvector}, it is chosen to ensure (via the Perron-Frobenius theorem, assuming that $A$ is irreducible) that the ranking $C_j^{(ev)}$ remains strictly positive. It has been shown by \citet{bonacich1987} that
\begin{align}
	\mathbf{v}_j\propto\sum_i\sum_{n=1}^\infty \lambda^{1-n}(A^n)_{ij}
\end{align}
where $\lambda$ is the maximum (principle) eigenvalue. As $(A^n)_{ij}$ represents the number of walks of length $n$ between vertices $i$ and $j$, it can be seen that the eigenvector centrality performs walks of \textit{all} lengths, weighted inversely by length, from each node. It should be noted that this result allows us to draw interesting parallels with the CTQW; consider the CTQW time-evolution operator for infinitesimal time $dt$:
\begin{align}
	U(t)=e^{-iHdt}=\sum_{n=0}^\infty \frac{1}{n!}(-i dt)^n H^n
\end{align}
That is, like the eigenvector centrality, the CTQW performs walks of all lengths at each infinitesimal time-step $dt$, weighted inversely by walk length. Thus, at the very least, the CTQW may provide the means for an eigenvector-like quantum centrality measure.

Unfortunately, the classical eigenvector centrality can provide ineffectual results when applied to directed graphs, for a multitude of reasons. For example, vertices not in a strongly connected component will be assigned an eigenvector vertex centrality of 0, as $A$ is no longer irreducible and the Perron-Frobenius theorem is no longer valid. Furthermore, `dangling nodes' (vertices or components with in-degree and zero out-degree) found in directed acyclic subgraphs can result in the eigenvector centrality `localising' or accumulating at the affected vertices \cite{perra2008,martin2014}. Furthermore, in the case of directed acyclic graphs, the eigenvector centrality provides no useful centrality information whatsoever --- the centrality measure of every vertex is identically assigned to be zero.

To rectify these issues, a wide range of variations to the eigenvector centrality have been proposed, including PageRank and Katz centrality. Of these, PageRank is arguably the most well known spectral centrality measure, due to its use in the Google search engine \cite{brin1998}. Designed to take into account directed networks, the PageRank improves on the eigenvector centrality by modifying the adjacency matrix to ensure stochasticity and irreducibility, whilst also introducing a `random surfer effect' --- a non-zero probability $1-\alpha$ (where $\alpha\in[0,1]$) that a walker at a vertex can transition to \textit{any} other adjacent or non-adjacent vertex. In practise, $\alpha$ is generally chosen to be 0.85, providing a good compromise between information flow via hyperlinks and the random surfer effect. Note that, whilst a modification of the eigenvector centrality, the PageRank's use of a stochastic transition matrix consequently allows it to be modelled in terms of a random walk of length one at each time-step.

Due to the close relationship between centrality measures and classical measures, directed quantum walks provide a natural starting point for exploring quantum centrality algorithms. For example, the quantum PageRank algorithm  utilises the discrete-time Szegedy quantum walk \cite{paparo2013,loke2017}, whilst the quantum stochastic walk (QSW) makes use of its hybrid classical-quantum regime to rank centrality of directed graphs \cite{loke2017,sinayskiy2013}. However, both approaches have their drawbacks; both require expanding the Hilbert space beyond that required for a standard quantum walk, hindering practical implementation, with additional classical decoherence suppressing quantum behaviour in the QSW. The CTQW also naturally lends itself to centrality analysis, and may provide a quantum analogue to the eigenvector centrality --- minus issues due to non-stochasticity and non-irreducibility due to its quantum propagation. Unfortunately, as the CTQW is not defined for directed/non-Hermitian graphs, possible applications regarding spectral centrality analysis has not been fully explored in the literature. Fortuitously, the pseudo-Hermitian $\eta$-CTQW extends the CTQW formalism to directed graphs, whilst preserving the eigenspectrum of the original graph.

In this section, we propose a quantum centrality measure based on the pseudo-Hermitian $\eta$-CTQW, which has the advantage of a purely quantum propagation and smaller statespace than both the Quantum PageRank and QSW. 

\subsection{Proposed quantum scheme and examples}
Taking a similar approach to the Quantum PageRank, we propose the following measure for assigning a centrality value to vertex $j$ on a directed graph represented by the pseudo-Hermitian Hamiltonian $H$:
\begin{align}
v_j=\lim_{T\rightarrow\infty}\frac{1}{T}\int_0^T\left|\bra{j}e^{-i\eta H\eta^{-1}t}\left(\frac{1}{\sqrt{N}}\sum_k\ket{k}\right)\right|^2 dt
\end{align}
Note that this method simply requires a Hilbert space of dimension $N$. This scheme will then be applied to the 3-vertex graph discussed previously.

Firstly, consider the non-probability conserving standard CTQW walk in the previous section. As the pseudo-Hermitian Hamiltonian preserves the graph eigenspectra, we expect the resulting dynamics to reflect properties that can be classically extracted from the spectra, such as vertex centrality. Thus, one potential way of extracting this information is simply to calculate the time-average of the probability on each vertex, and normalise it by the total probability time average.

Solving the Schr\"odinger equation for the unmodified Hamiltonian evolving the initial state $\ket{\psi(0)}=\frac{1}{\sqrt{3}}(\ket{1}+\ket{2}+\ket{3})$, we get the exact solution
\begin{align}
	\ket{\psi(t)} = \frac{1}{\sqrt{6}}\sqrt{5-3\cos(2t)}\left(\ket{1} + \ket{2}\right)+\frac{1}{\sqrt{3}}\ket{3}
\end{align}
where 
\begin{align}
	\braket{\psi(t)|\psi(t)} = 2-\cos(2t).
\end{align}
The total probability has a time average over $t=[0,\pi]$ (one period) of
\begin{align}
	\int_{0}^\pi\braket{\psi(t)|\psi(t)}dt = 2\pi,
\end{align}
therefore, calculating the time average of this result over $t=[0,\pi]$: 
\begin{align}
	\frac{1}{2\pi}\int_{0}^\pi |\braket{j|\psi(t)}|^2dt = \frac{5}{12}\braket{j|1} + \frac{5}{12}\braket{j|2} + \frac{1}{6}\braket{j|3}
\end{align}
That is, by this centrality measure, vertices 1 and 2 are ranked equal first, followed by vertex 3.

Next, consider the probability conserving $\eta$-CTQW case, $\tilde{U}(t) = \eta e^{-iHt} \eta^{-1}$. As we have an exact representation of $\eta$, we are also able to solve for $\ket{\tilde{\psi}(t)}$ exactly:
\begin{align}
	\ket{\tilde{\psi}(t)} = &\frac{1}{\sqrt{243}}\sqrt{101-20\cos(2t)}\left(\ket{1} + \ket{2}\right)\notag\\
    & ~~~~~~~ +\frac{1}{\sqrt{243}}\sqrt{41+40\cos(2t)}\ket{3}
\end{align}
Calculating the time average,
\begin{align}
	\frac{1}{\pi}\int_{0}^\pi |\braket{j|\tilde{\psi}(t)}|^2 dt = \frac{101}{243}\delta_{j1} + \frac{101}{243}\delta_{j2} + \frac{41}{243}\delta_{j3}
\end{align}

The numerical values of the CTQW and $\eta$-CTQW centrality rankings have been tabulated in \autoref{tab:centrality}, alongside the classical PageRank ($\alpha=0.85$) and eigenvector centrality rankings. Note that the two CTQW rankings strongly agree with the classical PageRank and eigenvector rank. Furthermore, the numerical values of the two CTQW rankings only differ by a maximum of about $1.23\%$, indicating that the pseudo-Hermitian similarity transform preserves information regarding vertex centrality in this particular example. A general statistical analysis will be carried out in the subsequent section.

\begin{table}[t!]
\centering
\begin{tabular}{ccccc}
\toprule
 & Eigenvector & PageRank & CTQW     & $\eta$-CTQW \\ \colrule
1  & 0.5    & 0.475              & 0.416667 & 0.415638             \\
2  & 0.5    & 0.475              & 0.416667 & 0.415638             \\
3  & 0    & 0.05               & 0.166667 & 0.168724             \\
\botrule
\end{tabular}
\caption{Centrality ranking of the vertices of 3-vertex graph \autoref{fig:3vertex}, using the classical PageRank method, the non-probability conserving pseudo-Hermitian CTQW, and the probability conserving pseudo-Hermitian CTQW ($\eta$-CTQW)}
\label{tab:centrality}
\end{table}

\begin{table}[t]
\centering
\begin{tabular}{ccccc}
\toprule
 & Eigenvector & PageRank & CTQW     & $\eta$-CTQW \\ \colrule
1   & 0.292893   & 0.25    & 0.386364 & 0.339192             \\
2   & 0.207107   & 0.25    & 0.113636 & 0.160808             \\
3   & 0.292893   & 0.25    & 0.386364 & 0.339192             \\
4   & 0.207107   & 0.25    & 0.113636 & 0.160808             \\
\botrule
\end{tabular}
\caption{Centrality ranking of the 4 vertices of graph \autoref{fig:4vertex}, using the classical PageRank method, the non-probability conserving pseudo-Hermitian CTQW, and the probability conserving pseudo-Hermitian CTQW ($\eta$-CTQW)}
\label{tab:centrality4}
\end{table}

\begin{figure}[t]
	\centering
    \includegraphics[scale=0.7]{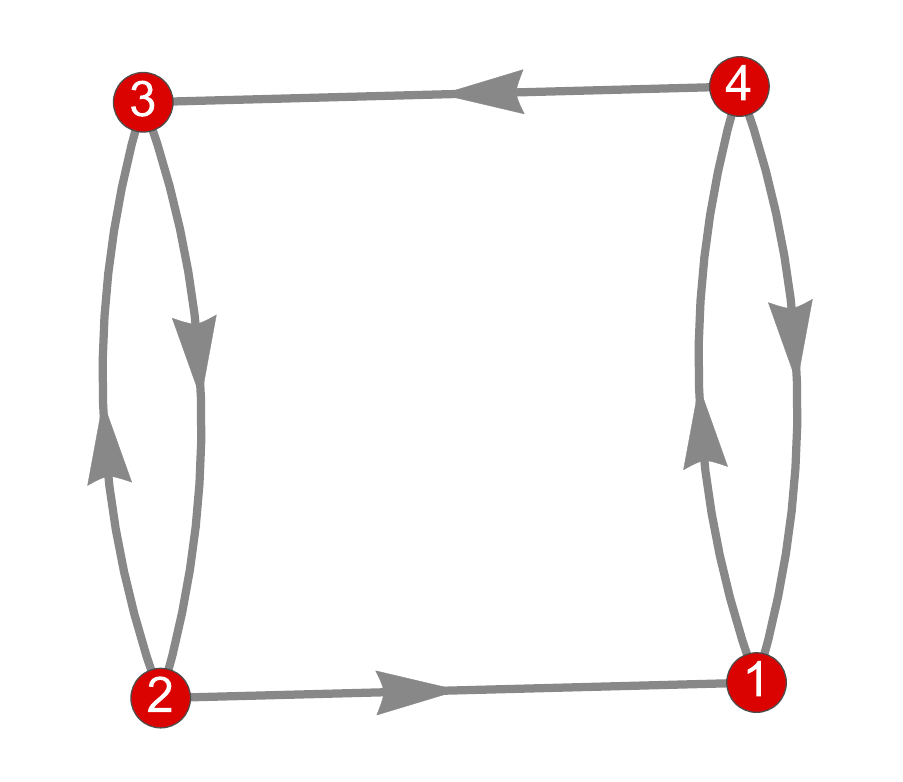}
    \caption{4-vertex directed graph}
    \label{fig:4vertex}
\end{figure}

Let's now consider a 4-vertex pseudo-Hermitian directed graph, as shown in \autoref{fig:4vertex}. Note that this graph is composed of two sets of equivalent vertices; vertices 1 and 3 (in-degree $=2$ and out-degree $=1$), and vertices 2 and 4 (in-degree $=1$ and out-degree $=2$).
Like the 3-vertex graph analysed previously, the CTQW probability can also be solved exactly, and thus the same method is applied to determine the vertex centrality.

The results of the CTQW and $\eta$-CTQW centrality test can be seen in \autoref{tab:centrality4}. In this case, both CTQW formulations and the eigenvector centrality ranked vertices 1 and 3 above vertices 2 and 4. Intuitively, this is perhaps expected, as vertices 1 and 3 have a greater in-degree and lower out-degree than vertices 2 and 4, resulting in walker probability accumulating on these two vertices. Interestingly, the classical PageRank measure does not distinguish between these two sets of vertices, as the limiting distribution of a (classical) random walk on this graph results in a unitary distribution. The agreement of the CTQW-based measures to the eigenvector centrality, as opposed to the PageRank, lends further credence to the suggestion that the CTQW measures centrality via a similar process to the eigenvector centrality.

As a final example, we briefly examined a pseudo-Hermitian interdependent network consisting of a 4-vertex directed graph and a 3-vertex directed graph, connected via complete interconnections (i.e. $B_0=J$). The graph and various centrality rankings of the vertices are shown in \autoref{fig:7vertex} --- it can be seen that the pseudo-Hermitian $\eta$-CTQW centrality ranking strongly agrees with the classical PageRank and eigenvector results, with the only disagreement involving the rankings of vertices 3 and 6, as well as 1 and 4 (both display degeneracy in the classical measures). Of note, the $\eta$-CTQW does a better job of ranking the vertices than the standard non-unitary CTQW in this case; perhaps indicating that the pseudo-Hermitian CTQW --- itself a mapping of a directed graph to an undirected, yet weighted, complete graph --- provides a better overall picture of the vertex ranks. In particular, the $\eta$-CTQW is the only ranking to break the top-ranked tie seen in the other measures, assigning slightly more importance to vertex 3 compared to vertex 6.
\begin{figure}[h!]
	\centering
	\subfigure{\includegraphics[scale=0.7]{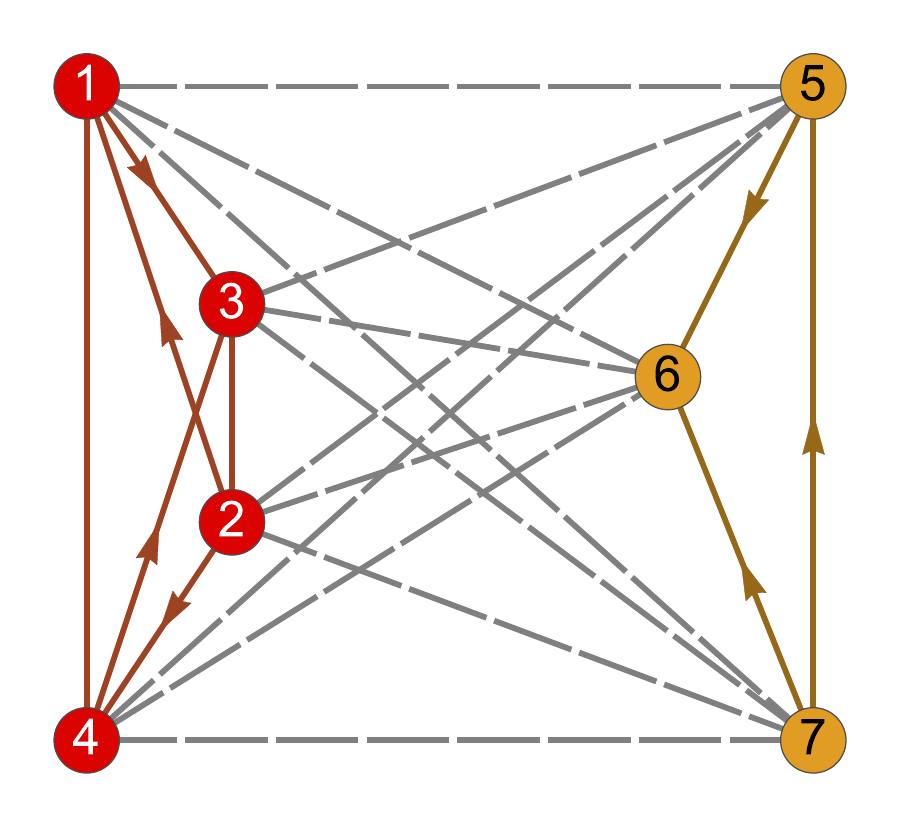}}
	\subfigure{\includegraphics[scale=0.9]{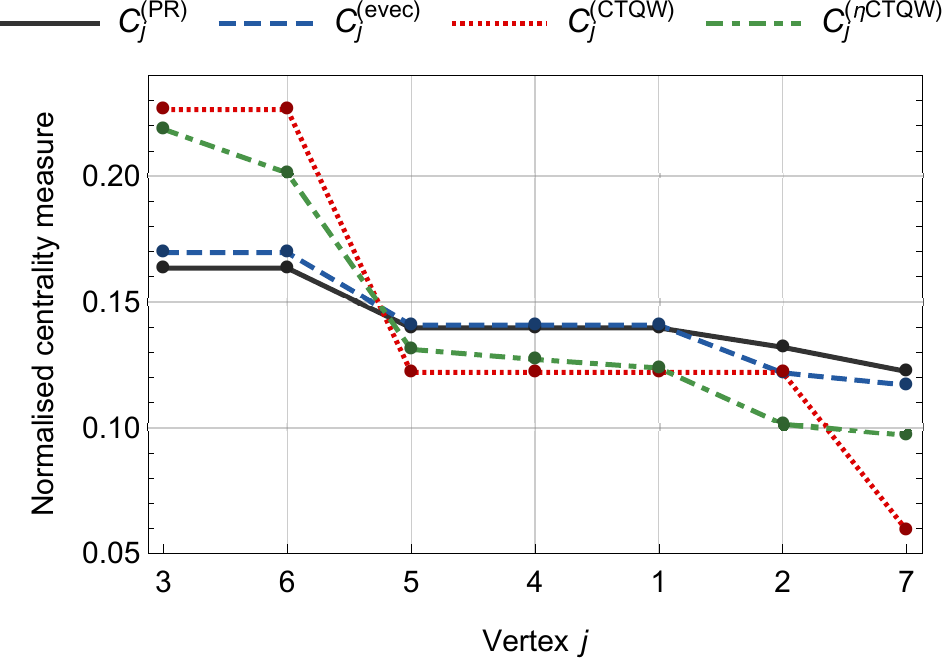}}
    \caption{\textbf{Top}: Interdependent network, consisting of a directed 4-vertex graph (red) connected to a directed 3-vertex graph (yellow) via complete interconnections (gray, dashed). \textbf{Above}: Centrality ranking of the vertices, ordered from highest ranking to lowest ranking vertex. Measures used include the classical PageRank (black, solid), the eigenvector centrality (blue, dashed), the standard CTQW (red, dotted) and the pseudo-Hermitian $\eta$-CTQW (green, dot-dashed).}
    \label{fig:7vertex}
\end{figure}

Another method of quantifying the correlation between the various centrality methods is to calculate their rank correlation coefficients. One such metric is Kendall's rank correlation coefficient, commonly referred to as Kendall's $\tau$ coefficient \cite{kendall1938}. By counting the number of pairwise disagreements between two ranked lists of length $N$, and dividing by normalization factor $\binom{N}{2}$, Kendall's $\tau$ may take values $-1\leq \tau\leq 1$, where $\tau=1$ denotes perfect agreements between the ranked lists, $\tau=0$ denotes no correlation, and $\tau=-1$ denotes perfect anticorrelation (i.e. one list is the reverse of the other). In the field of centrality analysis, Kendall's $\tau$ has become the definitive metric \cite{kim2012,chen2012,lulli2015,du2015}, by means of its ubiquity, efficient computability \cite{knight1966}, and the fact that variants exist that take into account ties \cite{kendall1945}. Despite this, Kendall's $\tau$ coefficient is not particularly suited towards comparing centrality measures. Measures with highly correlated top-ranked vertices may produce comparatively low $\tau$ values, as Kendall's $\tau$ equally weights all discordant pairs, regardless of where they appear in the ranking. Recently, weighted modifications have been proposed --- specifically catered to comparing centrality measures --- which use a hyperbolic weighting function to more heavily weight correlations of the top-ranked vertices. These include the AP (average precision) correlation \cite{yilmaz2008} and Vigna's $\tau$ correlation coefficient \cite{vigna2015}. As Vigna's $\tau$ further takes into account ties, we will apply Vigna's $\tau$ correlation coefficient to analyse the results of \autoref{fig:7vertex}.
\begin{table}
\centering
\begin{tabular}{r|cccc}
 & PageRank & Eigenvector & CTQW & $\eta$-CTQW\\
 \hline
PageRank & 1. & 0.912 & 0.937 & 0.896 \\
Eigenvector & 0.912 & 1. & 0.896 & 0.906 \\
CTQW & 0.937 & 0.896 & 1. & 0.761 \\
$\eta$-CTQW & 0.896 & 0.906 & 0.761 & 1. \\
\end{tabular}
\caption{Vigna's $\tau$ rank correlation coefficient compared for various classical (PageRank and eigenvector) and quantum (non-unitary CTQW and pseudo-Hermitian $\eta$-CTQW) centrality measures applied to the 7-vertex interdependent network in \autoref{fig:7vertex}.}
\label{tab:7vigna}
\end{table}

The results of this analysis can be seen in \autoref{tab:7vigna}. All compared centralities display very high correlation ($\tau\geq 0.8$), with the exception of the CTQW and $\eta$-CTQW, with a correlation value of $\tau=0.761$ (still a significant result). This is most likely due to the degeneracy seen in the CTQW ranking, which is completely broken in the $\eta$-CTQW; these resulting ties slightly lower the $\tau$ value, even though they do not cause disagreeing rankings per se. Finally, note that the $\eta$-CTQW achieves its highest correlation value with the eigenvector centrality at $\tau=0.906$, edging out correlation with the PageRank at $\tau=0.896$.

In these three examples, we have seen that the pseudo-Hermitian CTQW preserves the vertex centrality information from the original directed graphs, resulting in a vertex rank identical (barring broken degeneracy) to the classical PageRank. Whilst these relatively small examples allow us to verify the results of the centrality ranking by intuitively and qualitatively examining the graph structures by eye, this analysis is not sufficient to ensure that the centrality ranking proposed here generalises to other PT-symmetric graph structures. To do so, a statistical analysis featuring randomly generated directed graphs is required.

\subsection{Random directed networks}
To investigate the reliability of the pseudo-Hermitian CTQW on directed graphs, a statistical analysis will be undertaken using randomly generated directed networks. Here, we consider two classes of random networks --- Erd\H{o}s-R\'enyi networks, and scale-free networks.

A random Erd\H{o}s-R\'enyi graph, denoted $G(N,p)$, is comprised of $N$ vertices with edges randomly distributed via a Bernoulli distribution with probability $p$. For such a network, the vertex degree distribution $P(k)$ (the fraction of vertices with degree $k$) is binomial in form,
\begin{align}
	P(k)=\left(\begin{matrix}
		N-1\\
		k
	\end{matrix}\right)p^k(1-p)^{N-k-1}
\end{align}
resulting in most vertices with degree close to $np$, the mean number of connections \cite{erdoes1959,erds1960}. In order to produce a PT-symmetric directed graph satisfying the Erd\H{o}s-R\'enyi degree distribution, we take three approaches. Firstly, we generate numerous directed Erd\H{o}s-R\'enyi graphs (with parameters $N=15$, $p=0.3$) using the Python software package NetworkX \cite{hagberg2008}, and selecting from these 300 which satisfy pseudo-Hermiticity. An example is presented in \autoref{fig:ernx}, alongside a plot of the PageRank, non-unitary CTQW, and $\eta$-CTQW centrality measures for the pictured example. In this particular example, all three measures agree on the location of the top two ranked vertices, with slight discrepancies for the remaining vertices. Note that, from here onward, the classical eigenvector centrality is no longer included as a comparison, as we can no longer guarantee its performance --- a majority of the graphs in this and subsequent ensembles contain acyclic and non-strongly connected components that result in an eigenvector centrality value of zero.

%
%
\begin{figure}
    \centering
    \subfigure{\includegraphics[scale=0.65]{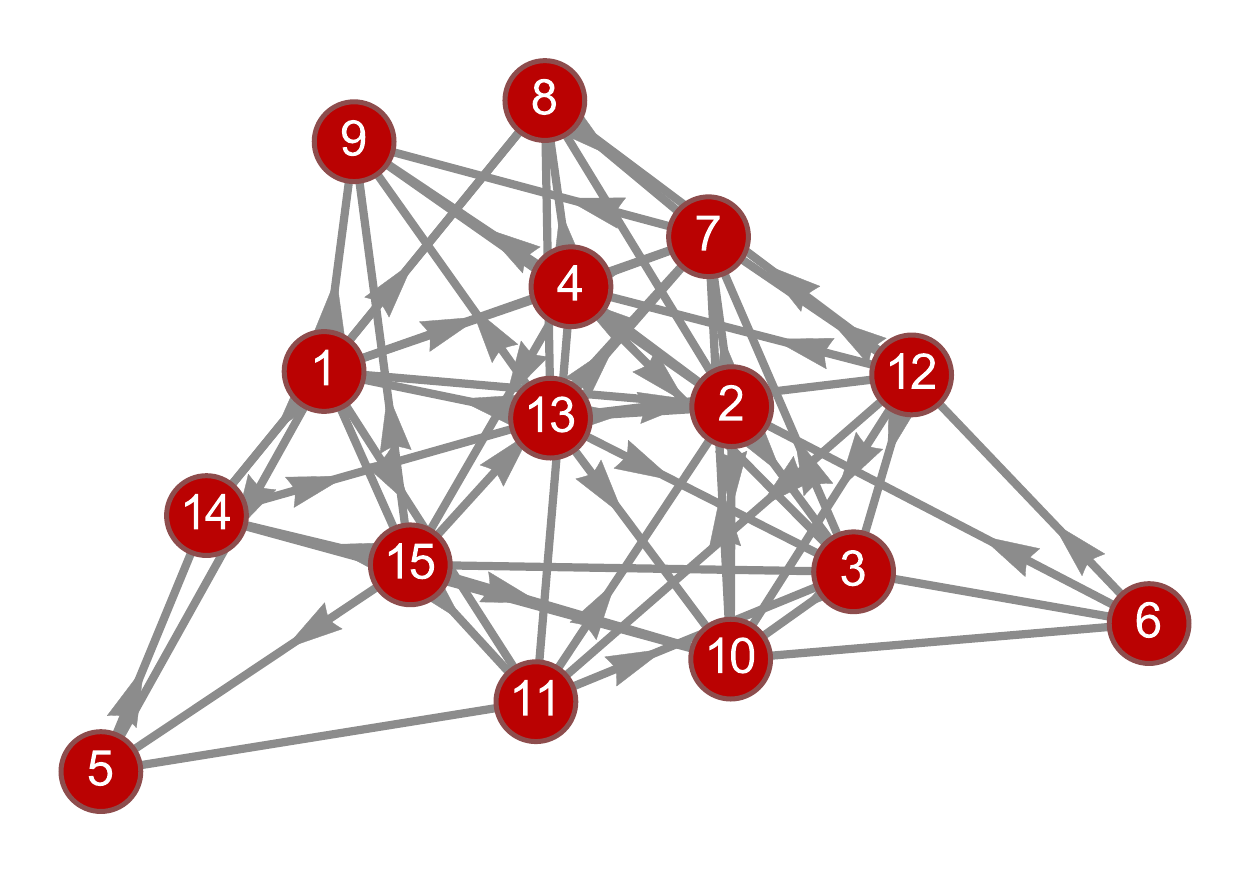}}
    \subfigure{\includegraphics[scale=0.7]{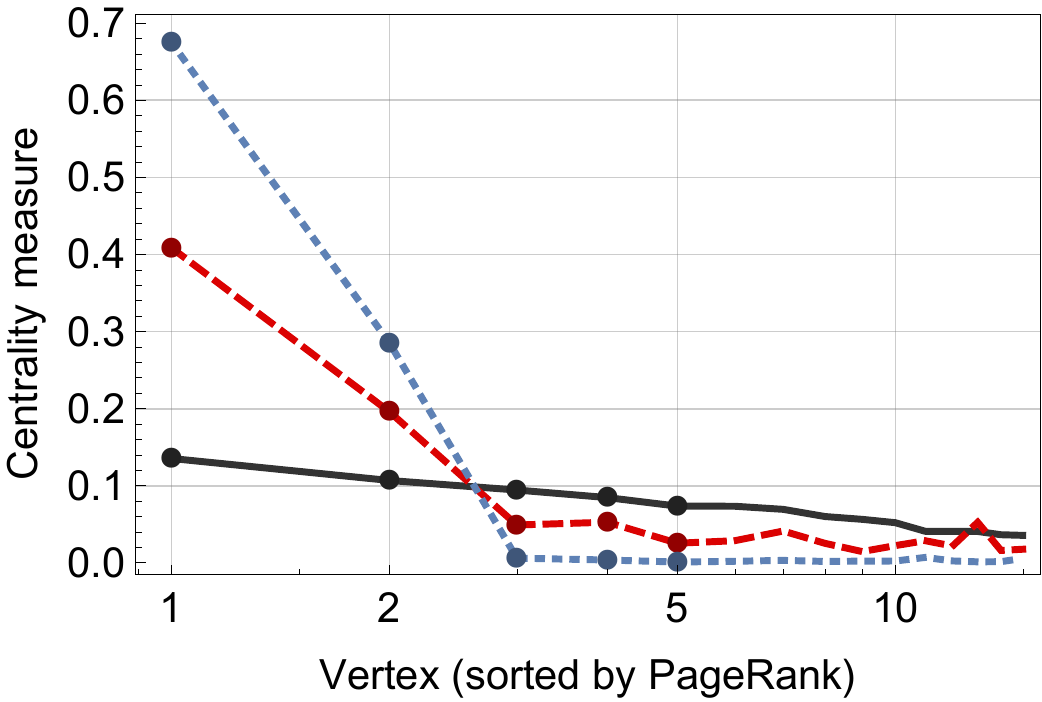}}
    \caption{\textbf{Top:} Randomly generated pseudo-Hermitian Erd\H{o}s-R\'enyi graph $G(15,0.3)$, with bidirectional edges allowed. \textbf{Above:} Centrality ranking of the vertices, ordered from most central to least central as per the classical PageRank (black). This is compared to the non-unitary CTQW (blue, dotted) and pseudo-Hermitian $\eta$-CTQW (red, dashed).}
    \label{fig:ernx}
\end{figure}

Our second approach to generating pseudo-Hermitian directed Erd\H{o}s-R\'enyi networks was motivated by computational constraints with using NetworkX, and a desire to generate larger pseudo-Hermitian graphs in a slightly more systematic way. Here, we first create an undirected graph of $N$ vertices with edges given by Bernoulli distribution with probability $p$. We then upper triangulize the resulting adjacency matrix, by setting everything below the diagonal to zero; in effect, imbuing direction on every edge in a systematic fashion. By restricting the adjacency matrix to be triangular, it is trivial to see that the Hamiltonian will also be triangular --- with eigenvalues given by the diagonal elements of $H$, the set of vertex in-degrees:
\begin{align}\label{eq:eigsin}
	\lambda = \{\text{deg}^{-}(v_i)~|~i=1,\dots,N\}
\end{align}
where $\text{deg}^-(v_i)$ is a function returning the in-degree of vertex $v_i$. As such, we ensure a real eigenspectrum, and simply restrict our random graph generator to output graphs with diagonalisable Hamiltonians in order to guarantee pseudo-Hermiticity. Note that, as the adjacency matrix is triangular, all graphs in this ensemble are directed acyclic graphs, and thus the classical eigenvector centrality no longer produces useful results (it assigns all vertices a centrality measure of zero).

%
%
\begin{figure}
	\centering
	\subfigure{\includegraphics[scale=0.65]{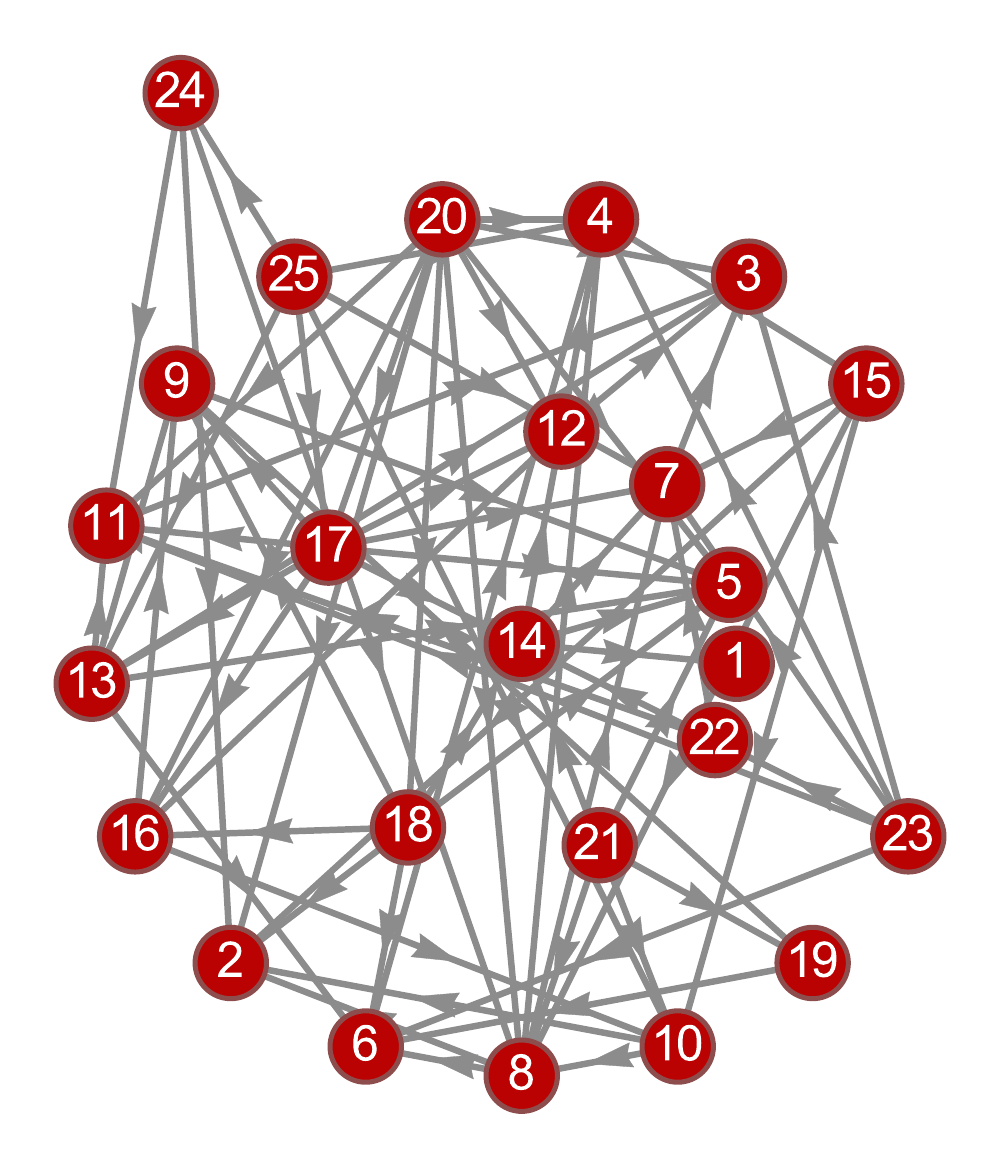}}
	\subfigure{\includegraphics[scale=0.7]{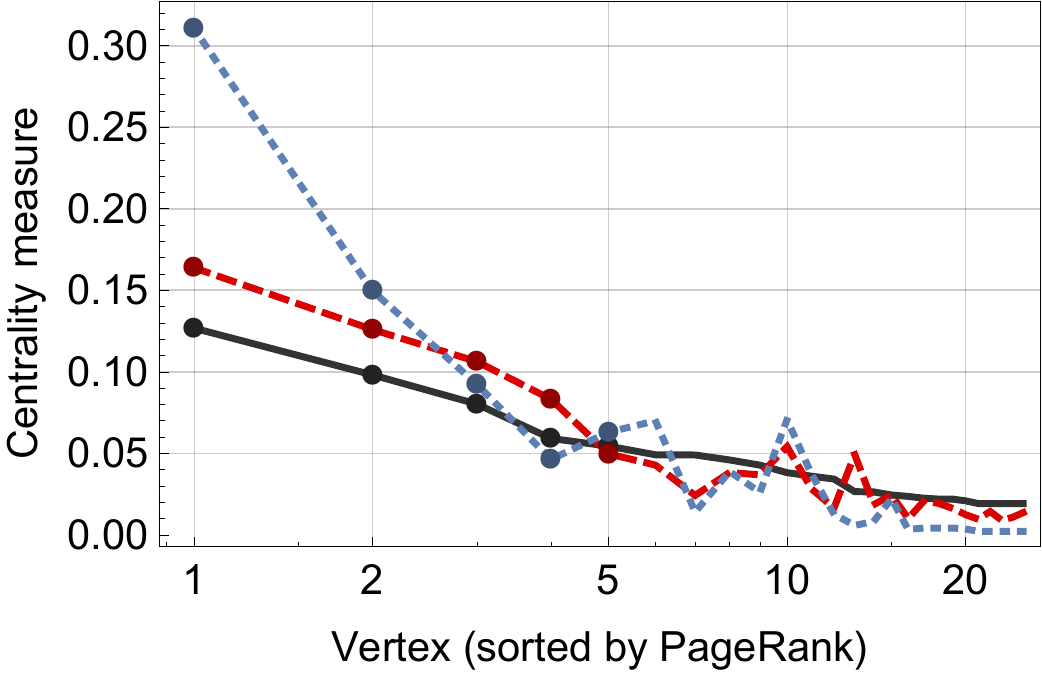}}
	\caption{\textbf{Top:} Randomly generated directed pseudo-Hermitian Erd\H{o}s-R\'enyi graph $G(25,0.3)$, where every edge is directed. \textbf{Above:} Centrality ranking of the vertices, ordered from most central to least central as per the classical PageRank (black).This is compared to the non-unitary CTQW (blue, dotted) and pseudo-Hermitian $\eta$-CTQW (red, dashed).}
	\label{fig:er3}
\end{figure}
An example of a randomly generated pseudo-Hermitian directed Erd\H{o}s-R\'enyi graph using this method is shown in \autoref{fig:er3}, generated with parameters $N=25$ and $p=0.3$, alongside the results of the PageRank and $\eta$-CTQW centrality measures. It can be seen that the $\eta$-CTQW and the PageRank strongly agree on the relative vertex rankings --- identically ranking the top four most central vertices, and satisfying the same general trend thereafter. This indicates that the $\eta$-CTQW continues to yield an admissible vertex centrality measure for larger, randomly generated graphs than in the previous section. 

In the above two methods of pseudo-Hermitian Erd\H{o}s-R\'enyi network generation, we are able to generate random graphs with bidirectional edges (allowing information to flow cyclically) and directed acyclic graphs respectively. Whilst only the former will permit use of the eigenvector centrality, the $\eta$-CTQW centrality algorithm may provide a usable centrality measure over both classes. To get a better understanding of how the $\eta$-CTQW centrality measure behaves over directed graph structures, we therefore introduce a third method of pseudo-Hermitian Erd\H{o}s-R\'enyi network generation, an intermediary between the two previously discussed classes. Here, we generate  pseudo-Hermitian Erd\H{o}s-R\'enyi networks as per our second (directed acyclic) approach, before introducing one bidrectional edge to the structure. Due to this addition, the overall graph is no longer directed acyclic, however all but one vertex form a directed acyclic subgraph.

Scale-free networks, compared to Erd\H{o}s-R\'enyi networks, exhibit a power law degree distribution of the form $P(k)\sim k^{-\gamma}$, due to a few very strongly connected vertices or `hubs' --- with a majority of vertices in the structure having significantly lower degree \cite{barabasi1999,barabasi2000}. As such, this makes them well suited to modelling a wide array of physical systems and networks with similar characteristics, for example power grids, the World Wide Web, social networks, and biochemical molecules \cite{albert2002,song2005}. To generate random pseudo-Hermitian scale-free graphs, we make use of the directed Barab\'asi-Albert algorithm: at each time-step, a vertex with $m$ directed edges is introduced to the system, and preferentially attached to existing vertices with higher degrees (with probability of being connected to vertex $i$ given by $p_i=k_i/\sum_{j}k_j$). This process continues until we have a graph containing the required number of vertices. 

A fortunate side effect of the directed Barab\'asi-Albert algorithm is that if we choose all $m$ edges introduced with each additional vertex to be \textit{inward}-pointing edges (resulting in $\text{deg}^{-}(v_i)=m~\forall i$), then the graph is necessarily lower triangular, leading to a Hamiltonian with real eigenspectrum as given by \autoref{eq:eigsin}. Similarly, if we choose all $m$ edges introduced with each vertex to be \textit{outward}-pointing edges (resulting in $\text{deg}^{+}(v_i)=m~\forall i$), the Hamiltonian will be upper-triangular and \autoref{eq:eigsin} continues to hold. Thus, like the Erd\H{o}s-R\'enyi case described previously, to ensure pseudo-Hermiticity we simply ensure the resulting randomly generated scale-free Hamiltonian is diagonalizable. As before, the directed Barab\'asi-Albert algorithm leads to the generation of directed acyclic graphs.
\begin{figure}
	\centering
	\subfigure{\includegraphics[scale=0.65]{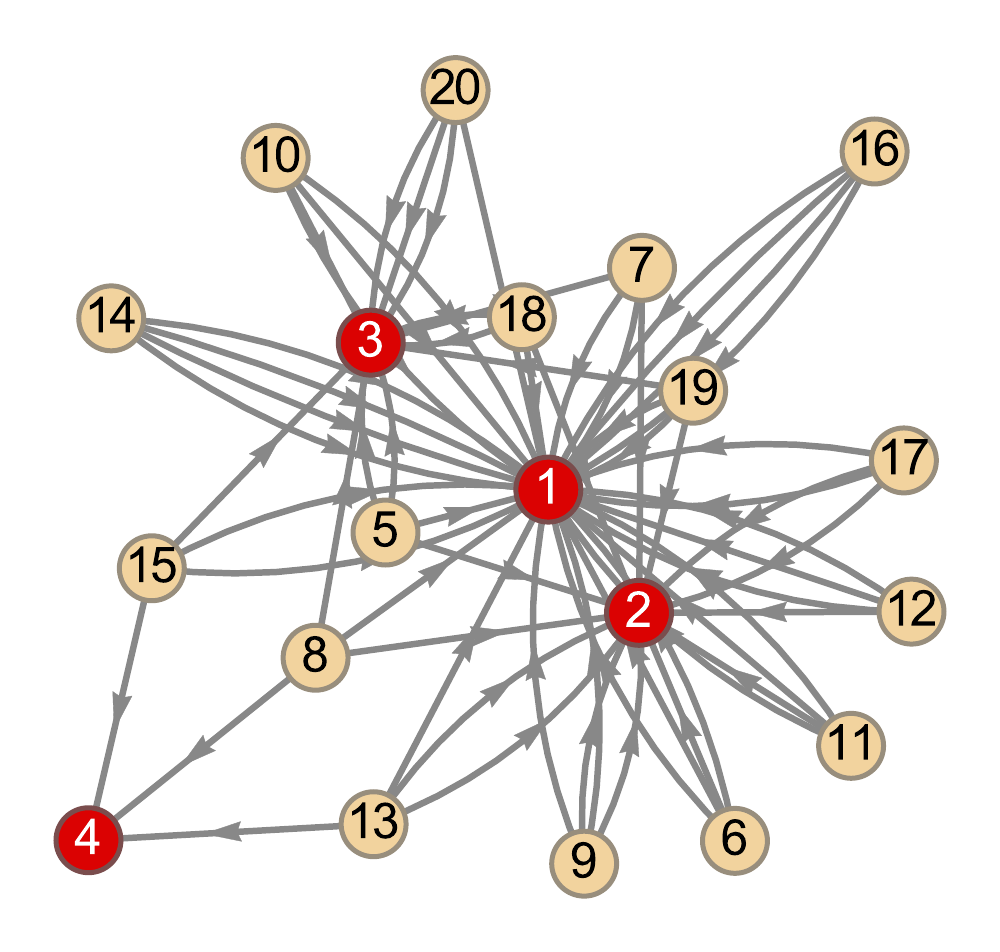}}
	\subfigure{\includegraphics[scale=0.7]{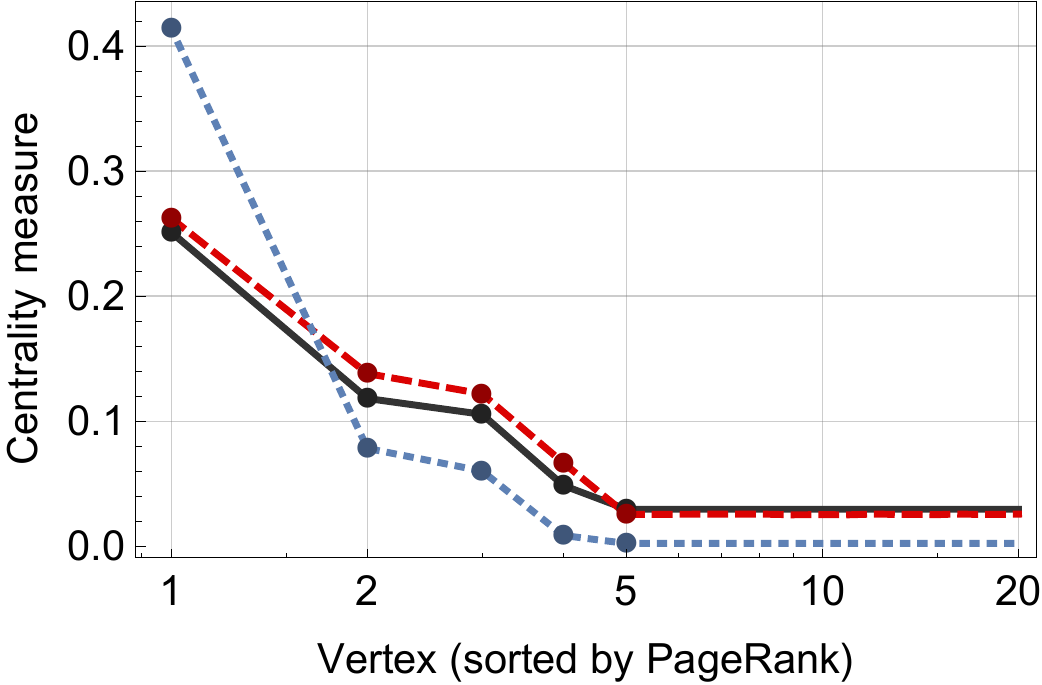}}
	\caption{\textbf{Top:} Randomly generated 20 vertex pseudo-Hermitian graph, with scale-free in-degree distribution where $m=4$. The more highly connected `hubs' (vertices with higher in-degree) are labelled in red. \textbf{Above:} Centrality ranking of the vertices, ordered from most central to least central as per the classical PageRank (black). This is compared to the non-unitary CTQW (blue, dotted) and the pseudo-Hermitian $\eta$-CTQW (red, dashed).}
	\label{fig:sf4}
\end{figure}
\begin{figure}
	\centering
	\subfigure{\includegraphics[scale=0.8]{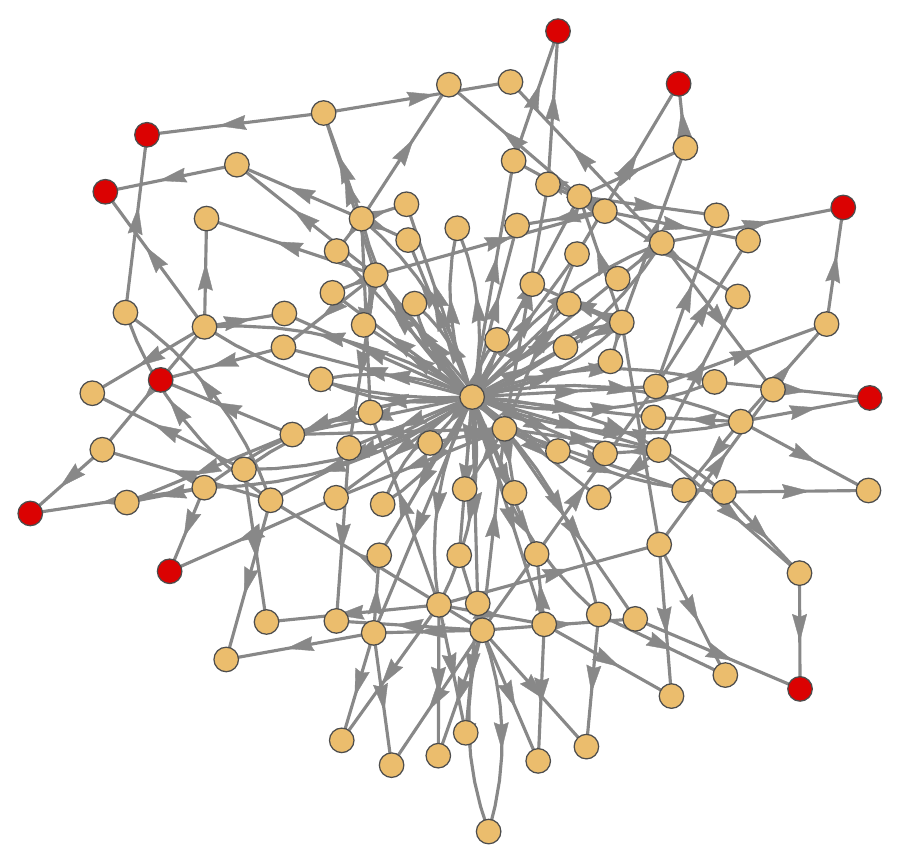}}
	\subfigure{\includegraphics[scale=0.7]{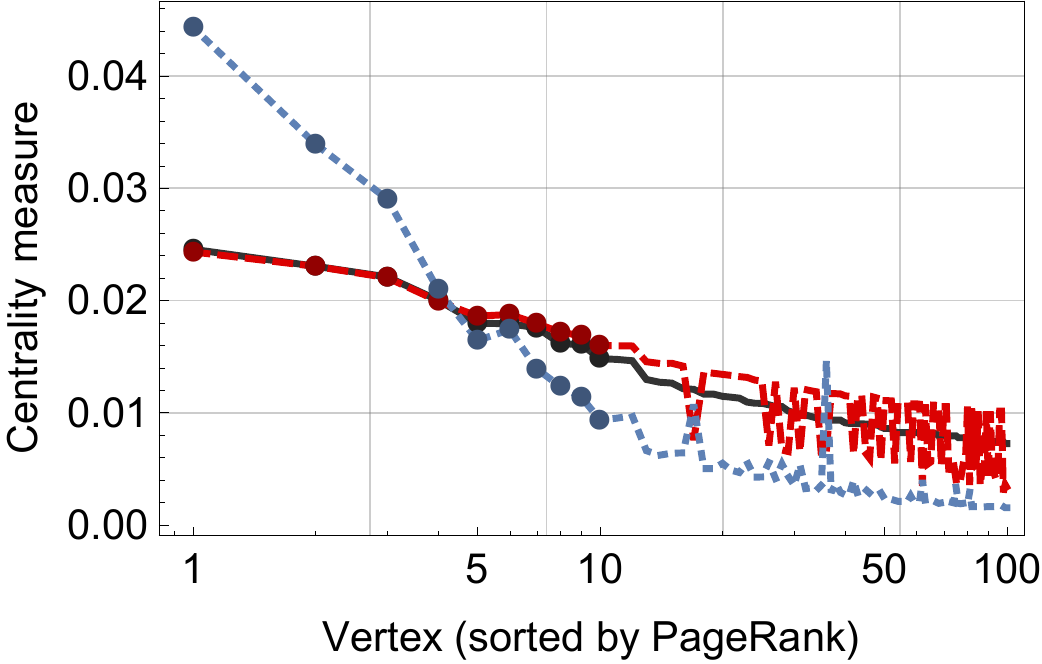}}
	\caption{\textbf{Top:} Randomly generated 100 vertex pseudo-Hermitian graph, with scale-free out-degree distribution where $m=3$. The top 10 vertices where network flow is likely to accumulate are labelled in red. \textbf{Above:} Centrality ranking of the vertices, ordered from most central to least central as per the classical PageRank (black). This is compared to the non-unitary CTQW (blue, dotted) and the pseudo-Hermitian $\eta$-CTQW (red, dashed).}
	\label{fig:100graph}
\end{figure}

\autoref{fig:sf4} shows a pseudo-Hermitian directed graph constructed via the Barab\'asi-Albert algorithm with parameters $N=25$, $m=4$, such that the in-degree vertex distribution is scale-free. By examining the classical PageRank and $\eta$-CTQW centrality measures, we see that they provide identical rankings for all 20 vertices, correctly picking out and ordering the four `hubs' (marked in red) with larger in-degree. Meanwhile, in \autoref{fig:100graph} we have a pseudo-Hermitian directed graph constructed via the Barab\'asi-Albert algorithm with parameters $N=100$, $m=3$, such that the \textit{out-degree} is scale-free, and constant in-degree of $3$. Again, the PageRank and $\eta$-CTQW display a high correlation, with the ranking of the top four most central vertices identical. However, in this case a subtlety must be addressed --- the PageRank algorithm is known to correlate with in-degree \cite{litvak2007,aiello2008}, as is the $\eta$-CTQW scheme by construction of the Hamiltonian in \autoref{eq:hamiltonian}. Hence, rather than assigning higher measures to vertices that are out-degree `hubs', both algorithms are preferentially selecting top-ranked vertices based on in-degree distribution. These correspond to the vertices at which the probability flow of a random walk is likely to accumulate after significant time.

\subsection{Statistical analysis}

\begin{figure*}
    \centering
    \subfigure[]{\label{subfig:ERNXavg}\includegraphics[scale=0.7]{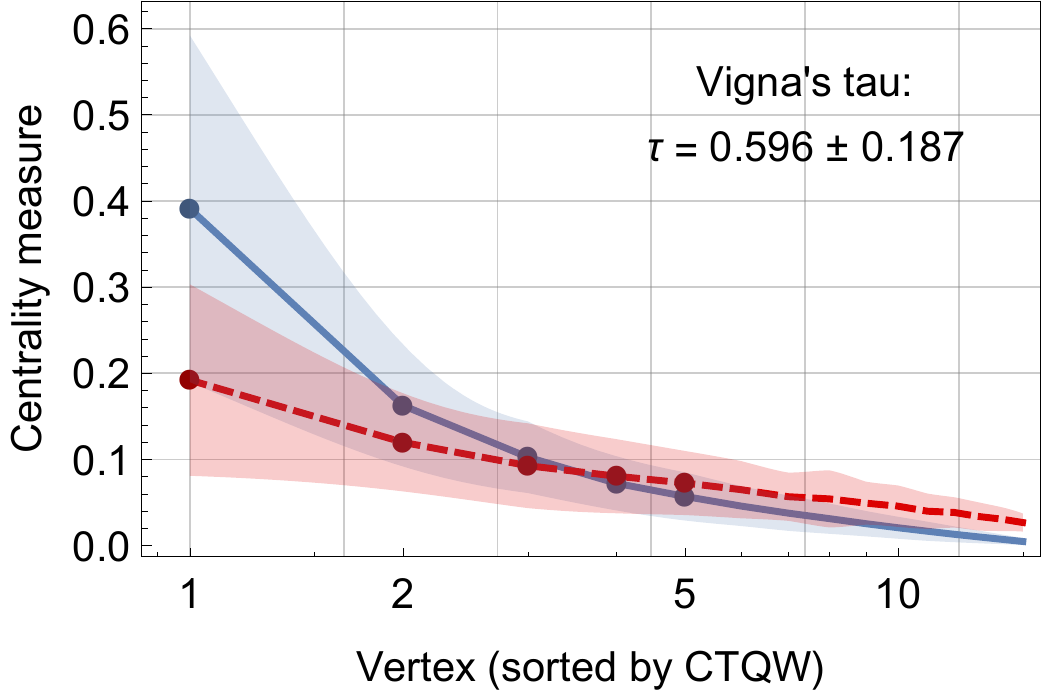}\hspace{1cm}\includegraphics[scale=0.7]{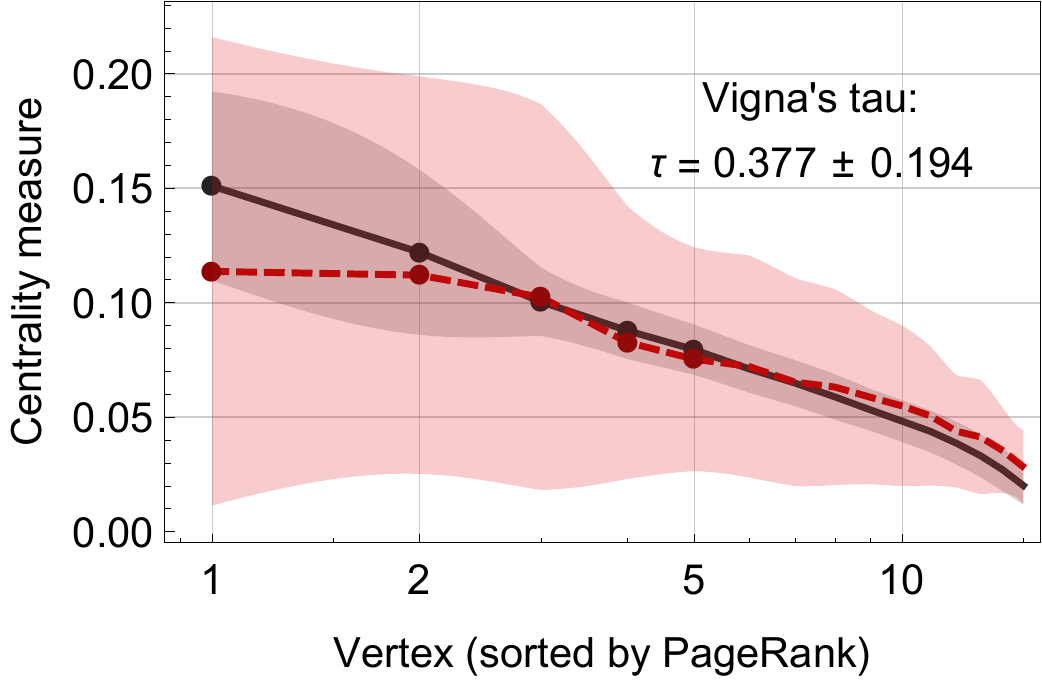}}
    \subfigure[]{\label{subfig:ERAddedDAGavg}\includegraphics[scale=0.7]{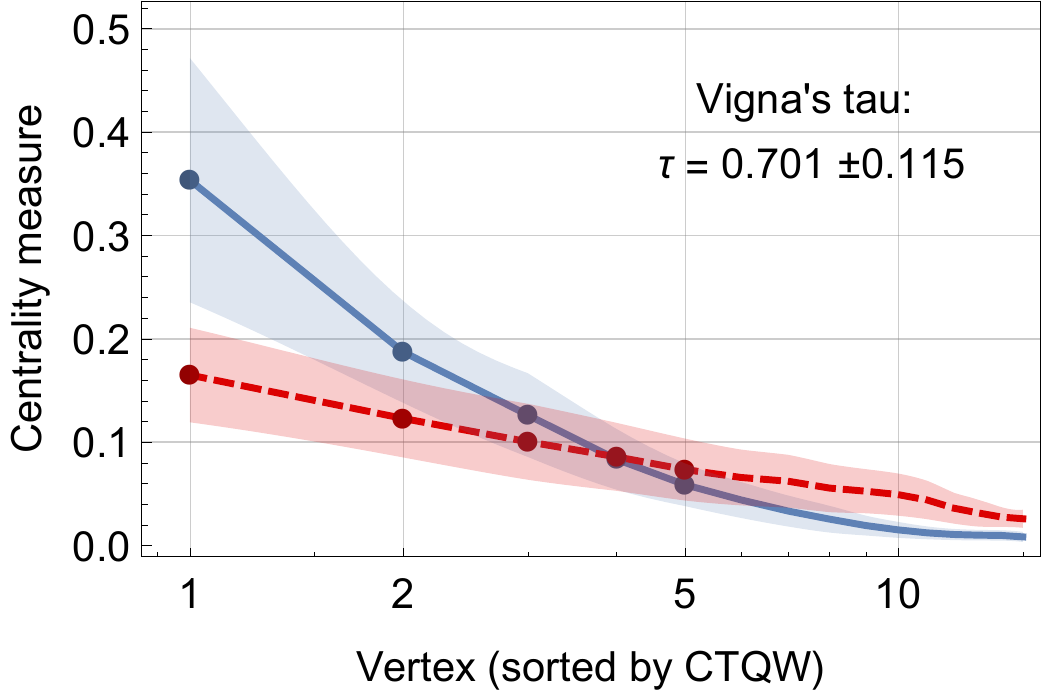}\hspace{1cm}\includegraphics[scale=0.7]{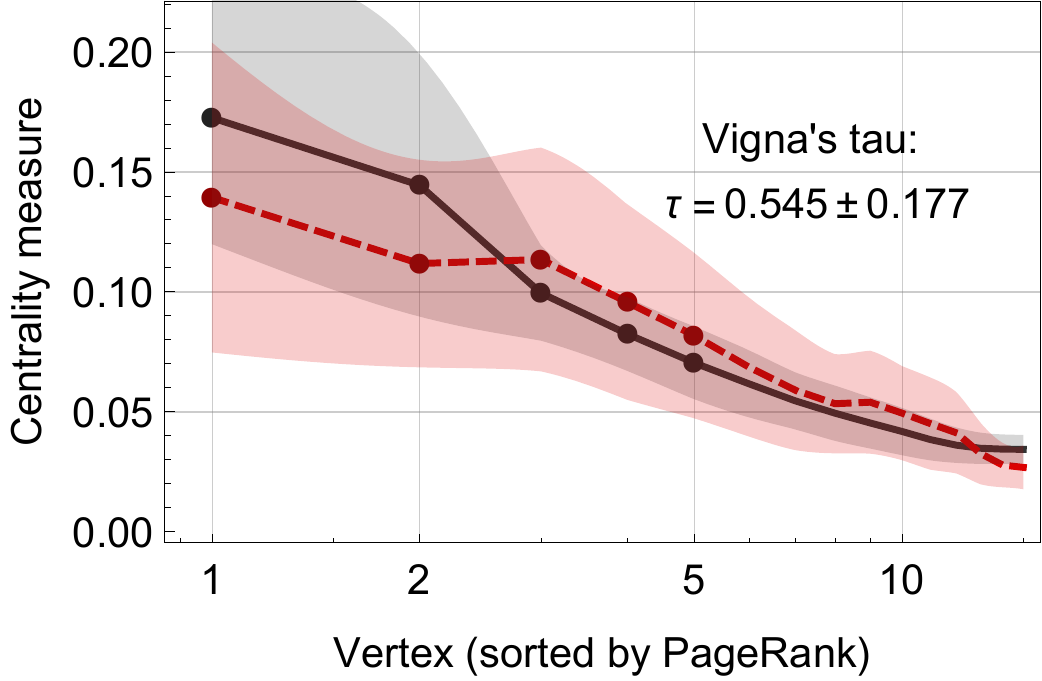}}
    \subfigure[]{\label{subfig:ER3avg}\includegraphics[scale=0.7]{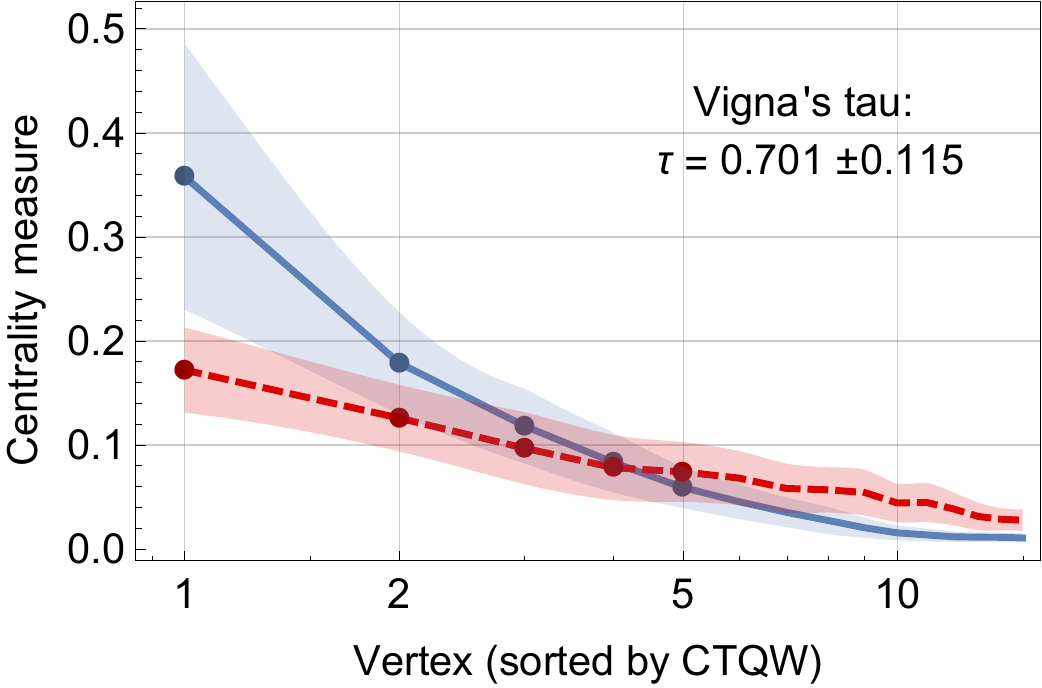}\hspace{1cm}\includegraphics[scale=0.7]{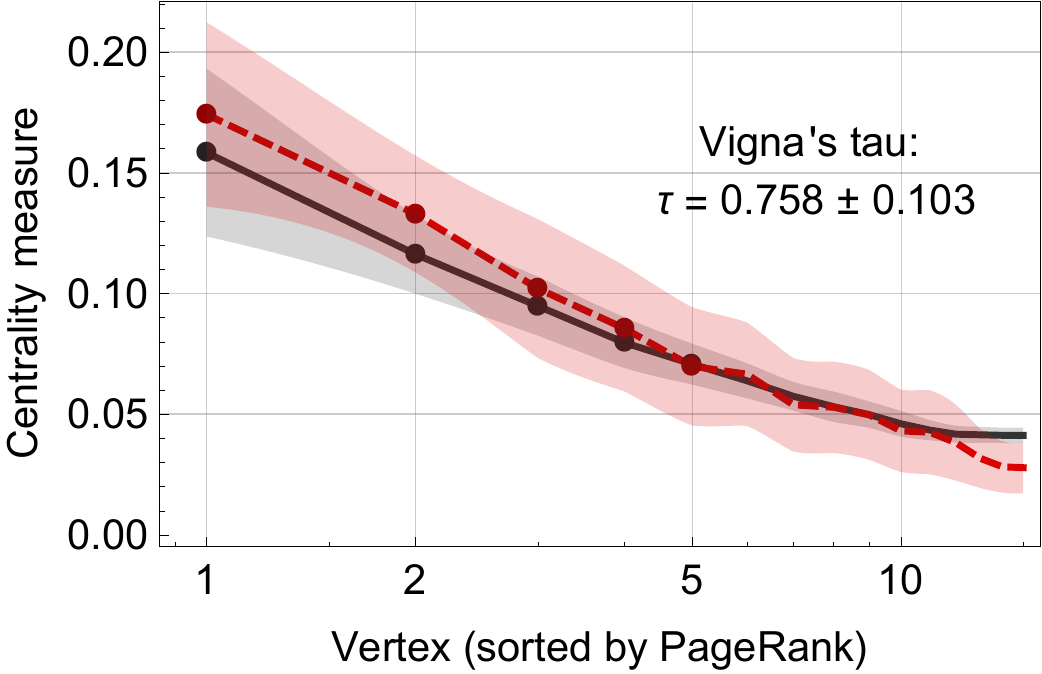}}
    \caption{Centrality measure values for the pseudo-Hermitian $\eta$-CTQW (red, dashed), compared against the non-unitary CTQW ranking (blue), and PageRank ranking (black), averaged over an ensemble of (a) 300 Erd\H{o}s-R\'enyi graphs with random bidirectional edges, (b) 100 Erd\H{o}s-R\'enyi graphs with one bidirectional edge, (c) 100 Erd\H{o}s-R\'enyi directed acyclic graphs}
    \label{fig:VCavgER}
\end{figure*}
\begin{figure*}
    \centering
    \subfigure[]{\label{subfig:ISavg}\includegraphics[scale=0.7]{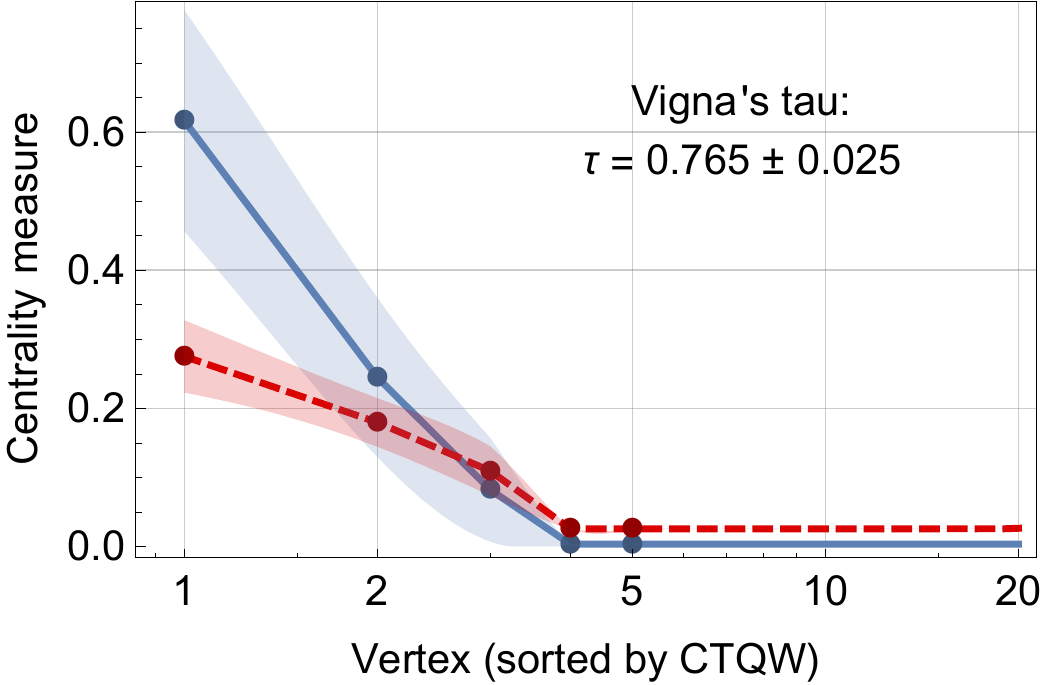}\hspace{1cm}\includegraphics[scale=0.7]{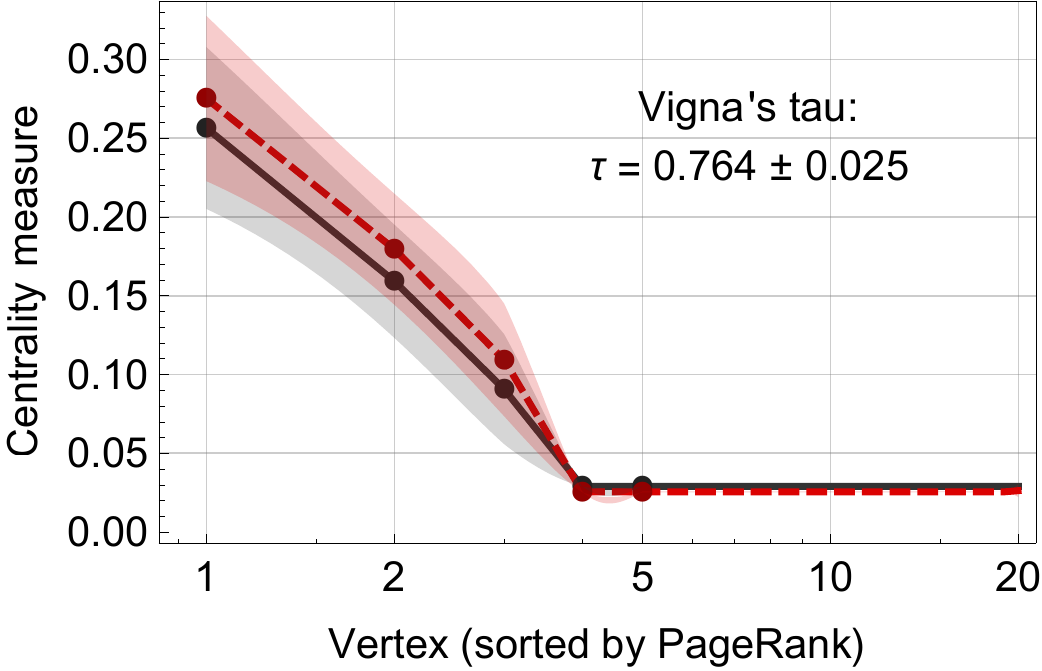}}
    \subfigure[]{\label{subfig:OSavg}\includegraphics[scale=0.7]{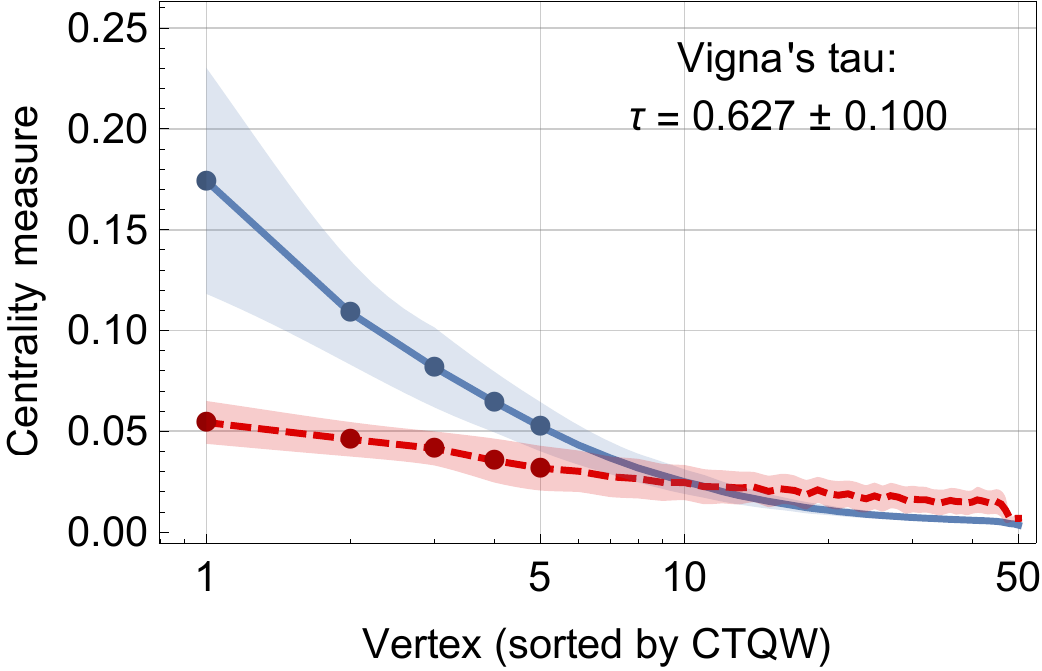}\hspace{1cm}\includegraphics[scale=0.7]{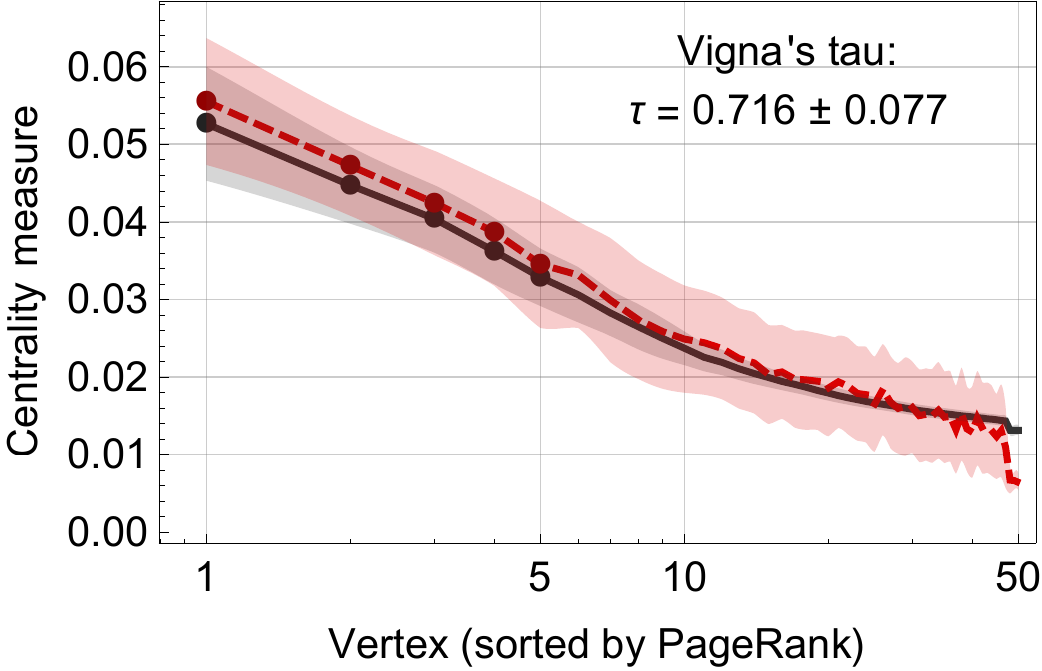}}
    \caption{Centrality measure values for the pseudo-Hermitian $\eta$-CTQW (red, dashed), compared against the non-unitary CTQW ranking (blue), and PageRank ranking (black), averaged over an ensemble of (a) 100 directed in-degree scale-free graphs, and (b) 100 directed out-degree scale-free graphs. The shaded areas represent the region within one standard deviation of the mean.}
    \label{fig:VCavgSF}
\end{figure*}
So far, we have considered particular Erd\H{o}s-R\'enyi and scale-free randomly generated graphs --- to explore how the pseudo-Hermitian CTQW centrality scheme behaves in general, it is pertinent to undertake a statistical analysis of an ensemble of random graphs. Ensembles of 300 random Erd\H{o}s-R\'enyi ($N=25$, $p=0.3$, random bidirectional edges permitted), 100 random Erd\H{o}s-R\'enyi ($N=25$, $p=0.3$, one bidirectional edge permitted), 100 random Erd\H{o}s-R\'enyi ($N=25$, $p=0.3$, directed acyclic), 100 random in-degree scale-free ($N=20$, $m=3$), and 100 random out-degree scale-free ($N=40$, $m=3$) were generated, and the PageRank, non-unitary CTQW, and $\eta$-CTQW vertex ranking determined for each graph. The mean and standard deviation of these centralities are plotted in \autoref{fig:VCavgER} (for Erd\H{o}s-R\'enyi ensembles) and \autoref{fig:VCavgSF} (for scale free ensembles), with the $\eta$-CTQW compared to both the non-unitary CTQW and classical PageRank. Furthermore, Vigna's $\tau$ rank correlation coefficient has been averaged across the ensemble, and is displayed on each plot.

Studying the results of \autoref{fig:VCavgER} and \autoref{fig:VCavgSF}, we may draw several conclusions. Firstly, the $\eta$-CTQW continues to reflect the directed structure of the network, agreeing with the non-unitary CTQW across all ensembles on the top 5 ranked vertices. This agreement is similarly reflected in Vigna's tau correlation coefficient, with $\tau\geq\sim0.6$ for every ensemble --- with lower values perhaps due to small discrepancies for lower ranked vertices.

Still, this statistical analysis has its drawbacks. The shaded areas, representing one standard deviation from the mean centrality values, indicate general ranking agreement across an ensemble \textit{only} when narrow enough and with a steep enough gradient such that each consecutive point, when moved upward/downward by one standard deviation, does not cause a swap in ranking (e.g. \autoref{subfig:ISavg}). Further, the converse is not true --- a large standard deviation does not imply a lack of agreement in ranking. In fact, two centrality measures could produce the exact same ranking across an entire ensemble, yet one measure might  simply have a greater variance in the values it assigns to the vertices. Similarly, Vigna's $\tau$ correlation coefficient, whilst a better indicator of overall rank agreement, continues to suffer from the fact that small discrepancies in ranking of lower-ranked vertices negatively affect the coefficient value. Thus, whilst these approaches might be useful in determining correlation between various centrality measures, they distract from the main question: how frequently do two centrality measures agree on the $k$ top-most ranked vertices?

In order to answer this quantitatively, we employ the Jaccard measure of set similarity \cite{maiya2010}. This provides an indicator of how well each centrality measure is able to determine the identity of the top $k$ highest centrality individuals. Firstly, for each graph, unordered sets containing the $n$ most central vertices according to each measure were compared --- the fraction of matching vertices providing a quantitative value for the agreement between the two measures. Finally, these were averaged over the entire ensemble, providing a general measure of the agreement between the PageRank and the $\eta$-CTQW, with uncertainty approximated by calculating the Agresti-Coull 95\% confidence interval \cite{agresti1998}.
\begin{figure}[h!]
	\centering
	\includegraphics[scale=0.9]{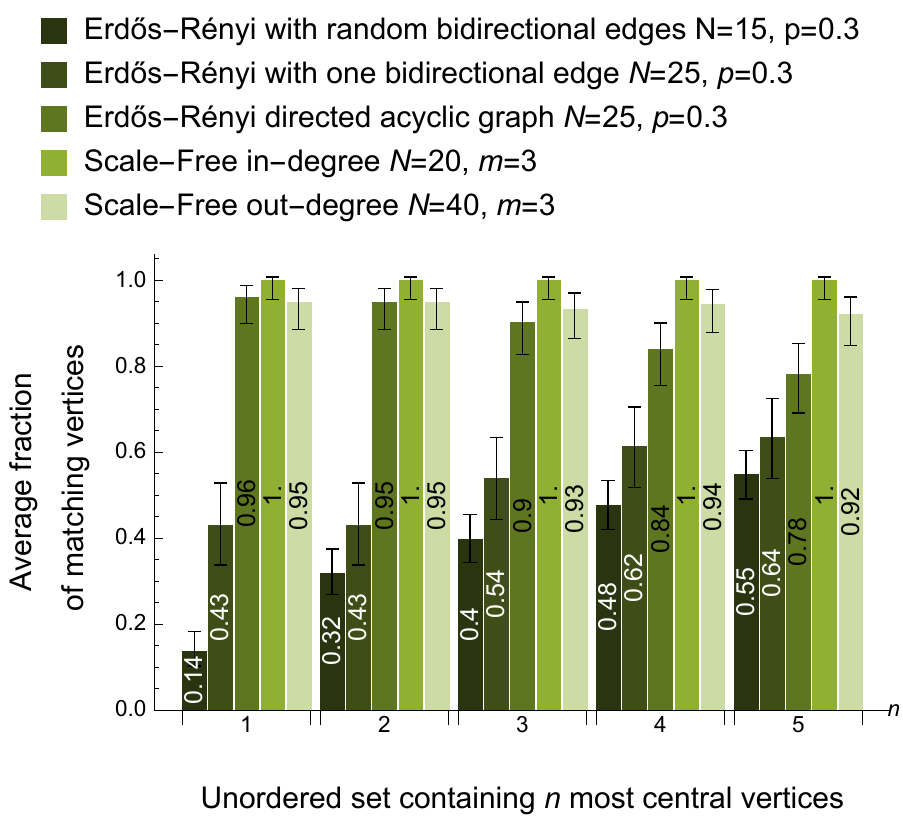}
	\caption{Chart showing the Jaccard set similarity between the classical PageRank algorithm and the $\eta$-CTQW centrality scheme for an ensemble of 300 directed Erd\H{o}s-R\'enyi graphs with bidirectional edges permitted, 100 directed Erd\H{o}s-R\'enyi graphs with unidirected edges, 100 directed in-degree scale-free graphs, and 100 directed out-degree scale-free graphs. Each bar represents the unordered set containing the $n$ most central vertices as determined by the PageRank and $\eta$-CTQW scheme, whilst the vertical axis gives the average fraction of matching vertices between these two sets. The error bars indicate the Agresti-Coull 95\% confidence interval.}
	\label{fig:fracmatch}
\end{figure}
The results of the statistical analysis are presented in \autoref{fig:fracmatch}. When considering just the most central vertex, the PageRank and $\eta$-CTQW are in excellent agreement in the case of the directed acyclic Erd\H{o}s-R\'enyi and scale-free ensembles, ranging from 95\% to 100\% agreement. As the number of vertices compared increases, there is a small decrease in the Jaccard set similarity, with all three ensembles of random graphs exhibiting agreement factors in the range of 90\% for the top two and three most central vertices. By the time we consider five vertices, scale-free networks retain an excellent agreement of 100\% and 92\%, whilst the directed acyclic Erd\H{o}s-R\'enyi ensemble exhibits a reasonably good agreement factor of 78\%. These trends can be partially explained by considering the behaviour of the degree distributions:

\begin{itemize}
\item Erd\H{o}s-R\'enyi networks, with a majority of vertices having degree close to the mean, generally results in the highest ranked vertices having similar centrality measures. As such, beyond the top three, small variations in the PageRank and $\eta$-CTQW vertex ordering appear, leading to discrepancies.
\item In-degree scale-free networks, with a small number of highly connected vertices, should easily distinguish these vertices (the `hubs') as most central to the network. Beyond the hubs, the power law characteristic results in the majority of remaining vertices having similar degree --- leading to small variations in vertex ordering, and thus the discrepancies observed between the PageRank and $\eta$-CTQW as more vertices are compared.
\end{itemize}

However, comparing the $\eta$-CTQW and PageRank for the non-directed acyclic Erd\H{o}s-R\'enyi ensembles (those with bidirectionality of edges permitted) we see a significant reduction in the agreement of the top 5 vertices. For instance, in \autoref{fig:fracmatch}, it can be seen that top-most vertex Jaccard set similarity between the PageRank and $\eta$-CTQW on the ensemble with one permitted bidirectional edge is 43\%; this drops to 13\% when random bidirectional edges are permitted. This could be due to a multitude of factors:
\begin{itemize}
	\item the PageRank might provide a significantly different rank to other classical measures, which the $\eta$-CTQW is more inclined to agree with --- this difference may be magnified on non-directed acyclic graphs;
	\item localisation of the $\eta$-CTQW may be occurring, due to either classical effects \cite{martin2014,pastor-satorras2016} or quantum effects (Anderson localisation).
\end{itemize}
Interestingly, the $\eta$-CTQW centrality measure appears to allows us to apply an eigenvector-like quantum centrality algorithm that agrees readily with the classical PageRank on directed acyclic graphs --- on which the eigenvector centrality provides inconclusive results --- whilst failing to agree with PageRank on non-directed acyclic graphs.

Ultimately, whatever the reason, further investigation is required to determine the likely cause of the discrepancy. Note that this is not a negative result per se --- depending on the model represented by the graph structure, the $\eta$-CTQW could be providing a better result of marking influential and central nodes. However, this analysis is beyond the scope of this paper, and is reserved for future research. Nevertheless, the results presented here show that the $\eta$-CTQW provides centrality rankings for several classes of randomly generated graphs that are consistent with the classical PageRank algorithm.

\section{conclusion}
\label{sec:conc}
In this paper, we have introduced and expanded a framework for continuous-time quantum walks on directed graphs, by utilising PT-symmetry. In the case of interdependent networks of directed graphs, a sufficient condition for ensuring PT-symmetry was detailed, and the directed walk formalism was shown to be equivalent to simulating a continuous-time quantum walker on an undirected, weighted, complete graph with self-loops. This may potentially lead to easily-implementable experimental directed continuous-time quantum walks.

Finally, we have introduced a quantum scheme for centrality testing on directed graphs, by utilising PT-symmetric continuous-time quantum walks --- unlike other directed quantum-walk based centrality-measures, our method does not require expanding the Hilbert space to ensure unitary behaviour. A statistical analysis was performed, confirming the CTQW centrality measure proposed here is consistent with classical centrality measures for various classes of randomly generated directed acyclic graphs.


Preliminary results on 4-vertex pseudo-Hermitian directed graphs have shown that the CTQW centrality ranking is able to distinguish non-equivalent sets of vertices that the classical PageRank cannot. This is likely due to the CTQW providing an eigenvector-like centrality measure in the quantum regime; calculating the rank correlation coefficients supports this interpretation. However, further work is required to fully understand the distinguishing power of the pseudo-Hermitian CTQW centrality measure.

Quantum walks remain an important physical tool, linking the fields of information theory, quantum computation, and complex quantum dynamical modelling. Following on from this work, we aim to utilise the PT-symmetric CTQW framework to model and simulate behviour in physical biochemical systems, such as electron or excitontransport. Future work will also involve exploring methods of implementing the PT-symmetric CTQW centrality scheme on physical systems.

\section{Acknowledgements}
The authors would like to thank Thomas Loke, Yogesh Joglekar, Ping Xu, and Yongping Zhang for valuable discussions regarding PT-symmetry and interdependent networks. J. A. Izaac would like to thank the Hackett foundation and The University of Western Australia for financial support.

\bibliography{PT}

\end{document}